\documentclass[11pt, letterpaper]{article}

\usepackage{graphicx}
\usepackage{amsmath}
\usepackage{amssymb}
\usepackage{setspace}
\usepackage{epstopdf}
\usepackage{color}
\usepackage[letterpaper,left=1in,right=1in,top=1in,bottom=1in]{geometry}
\usepackage[ruled, vlined, longend]{algorithm2e}
\usepackage{multirow}
\usepackage{longtable}
\usepackage{rotating}
\usepackage{threeparttable}
\usepackage{colortbl}
\usepackage{amsthm}
\usepackage[modulo]{lineno}
\usepackage{enumerate}

\newtheorem{remark}{Remark}

\newtheorem{proposition}{Proposition}

\newtheorem{theorem}{Theorem}

\newtheorem{definition}{Definition}
\newtheorem{assumption}{Assumption}
\newtheorem{example}{Example}

\usepackage{xpatch}
\makeatletter
\xpatchcmd{\@thm}{\thm@headpunct{.}}{\thm@headpunct{}}{}{}
\makeatother

\newcommand{\m}{\mathbb}

\title{\large \bf Computing control invariant sets of nonlinear systems: decomposition and distributed computing}

\author{
    \centerline{\normalsize Benjamin Decardi-Nelson$^{a}$, Jinfeng Liu$^{a,}$
    \thanks{Corresponding author: J. Liu. Tel: +1-780-492-1317. Fax: +1-780-492-2881. Email: jinfeng@ualberta.ca}}
    \vspace{5mm} \\
    \centerline{\small $^{a}$ Department of Chemical \& Materials Engineering, University of Alberta,}\\
    \centerline{\small Edmonton AB T6G 1H9, Canada}
}

\allowdisplaybreaks

\begin{document}

\date{}
\maketitle
\setstretch{1.39}

\begin{abstract} \label{sec:abstract}
	In this work, we present a distributed framework based on the graph algorithm for computing control invariant set for nonlinear cascade systems. The proposed algorithm exploits the structure of the interconnections within a process network. First, the overall system is decomposed into several subsystems with overlapping states. Second, the control invariant set for the subsystems are computed in a distributed manner. Finally, an approximation of the control invariant set for the overall system is reconstructed from the subsystem solutions and validated. We demonstrate the efficacy and convergence of the proposed method to the centralized graph-based algorithm using several numerical examples including a six dimensional continuous stirred tank reactor system.
\end{abstract}



\section{Introduction} \label{sec:introduction}

Many engineered systems such as chemical processes are usually designed to operate within a prescribed safe set -- usually in the form of state constraints. However, the ability of the state of these engineered systems to stay within the prescribed safe set is hindered by limits on the available control energy used to control these systems \cite{homer2017}. Thus, to avoid unsafe and catastrophic events during operation, it is very critical that the set of states for which a control input exist to keep the system's states within the state constraint at all times is known prior to operation or for controller design purposes. This set of states is known as a control invariant set. 
The concept of control invariant set has close connections with reachable sets \cite{mitchell2005}, null controllable sets \cite{homer2020} and viability kernels \cite{aubin2009}.

Control invariant sets have been found to be important in controller design and assessment \cite{blanchini1999}. In particular, it has been found to be useful as a terminal region constraint for model predictive control strategies \cite{mayne2001,cannon2003}. Despite the well known importance of control invariant sets, accurate and efficient determination of control invariant sets is still an open issue, even for linear systems. This is because determining a control invariant set for a given system involves considering all possible trajectories of the system. Unfortunately, computing all possible system trajectories become challenging when the system dimension is high, a phenomenon known as curse of dimensionality. For this reason, many current state-of-the-art algorithms have inadequate scalability when it comes to the dimension of the system states.

To this end, several algorithms and methods for computing control invariant sets have been proposed \cite{decardi2021,rakovic2010,rungger2017,fiacchini2010,bravo2005}. Linear systems which are usually based on the well-known dynamic programming method \cite{bertsekas1972} and polytopic computations, have in particular received considerable attention in literature \cite{rakovic2010,rungger2017}. However, these methods for linear time invariant (LTI) systems are based on vertex enumeration of the set's full-dimensional polytopic representation \cite{rungger2017}. The number of vertices in the polytope grows exponentially as the system's dimension increase, resulting in intractable computation. For nonlinear systems, only a few results exist with majority of the algorithms either being grid-based \cite{homer2017,mitchell2005,homer2018} or set-based \cite{decardi2021}. Similar issues which are considerably more challenging exist for nonlinear systems.

A trade-off between set complexity and computational complexity is frequently made to overcome scalability concerns in determining control invariant sets. The methods presented in \cite{kurzhanski2000,frehse2011} have been successful in computing control invariant set of high dimensional systems but assume a control invariant set of prescribed complexity or shape which results in faster computation but overly conservative sets. A decomposition and reconstruction based method for exact computation of backward reachable sets of self-contained nonlinear subsystems using the Halmilton-Jacobi (HJ) formulations was developed in \cite{chen2018}. In \cite{riverso2017,rakovic2010}, decentralized and distributed algorithms were developed to obtain a family of practical positively invariant sets of linear systems. In \cite{li2020}, a decomposition based method was proposed to compute the backward reachable sets of nonlinear systems. In these cases, however, the missing interconnection information were treated as disturbances which results in conservative results. 

In this paper, we focus on cascade systems and present a system decomposition method and a distributed approach for computing control invariant sets. The proposed algorithm exploits the structure of the interconnections within a process network and decomposes the entire process network into smaller subsystems. Following the decomposition, a distributed approach is developed to compute the control invariant set of the entire system. The proposed approach adopts graph-based algorithms in computing the control invariant sets \cite{decardi2021}. In contrast to other works on distributed computation of control invariant sets, our proposed approach produces sets that approximates the largest control invariant set since the missing interconnection information is not treated as disturbances. The proposed distributed computing algorithm may be extended to more general nonlinear systems. Some preliminary results of this work were reported in \cite{decardi2022ACC}. Compared with \cite{decardi2022ACC}, this paper provides significantly extended explanations and discussion, and significantly more results exploring various features of the proposed distributed computing approach. 

The remainder of the paper is organized as follows: Section~\ref{sec:problem_formulation} is concerned with preliminaries and problem formulation as well as a brief introduction to directed graph representation of dynamical system with controls and control invariant set investigation of the graph. In Section~\ref{sec:system_decomposition}, we present the system decomposition method with overlapping states and discuss the implications of computing control invariant set via system decomposition. In Section~\ref{sec:cis_computation}, we present an algorithm to compute the outer approximation of the largest control invariant set in a distributed manner. We also briefly discuss some computational issues in this section. We present a several examples in Section~\ref{sec:examples} to demonstrate the efficacy of the proposed method and then conclude the paper in Section~\ref{sec:concluding_remarks}. 

\section{Problem formulation and background} \label{sec:problem_formulation}

\subsection{Notation}

Throughout this article, the operator proj$_i(x)$ denotes the projection of the set or point $x$ onto the subspace of subsystem $i$. The operator $\slash$ denotes set subtraction such that $\m A \slash \m B = \{ x : x \in \m A, x \notin \m B \}$. $G = (V, E)$ represents a directed graph with $V$ denoting the set of vertices of the graph and $E$ denoting the set of ordered pairs of vertices known as edges. The operator $|\cdot|$ denotes the Euclidean norm of a vector.

\subsection{Problem formulation}

We are concerned with a class of nonlinear systems composed of $N$ subsystems coupled together in a cascade manner. 
\begin{figure}[tbp] 
    \center{\includegraphics[width=0.6\columnwidth]{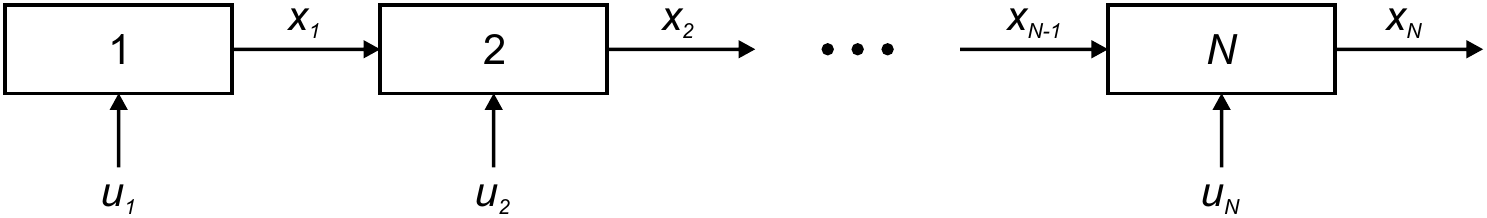}}
    \caption{\label{fig:cascade_system} A schematic diagram of a cascade process system}
\end{figure}
A schematic diagram of the cascade process system is shown in Figure~\ref{fig:cascade_system}. The dynamics of the overall system is described by
\begin{equation} \label{eqn:system}
    x^+ = f(x,u)
\end{equation}
where $x \in \m R^n$ is the current state of the system, $u \in \m R^m$ is the current control input, and $x^+ \in \m R^n$ denotes the state of the system at the next sampling time. We assume that the state and the control input of the system are restricted to be in the compact constraint sets $X \subset \m R^n$ and $U \subset \m R^m$ respectively. Without loss of generality, we also assume that the vector field $f:X \times U \rightarrow X$ is a sufficiently smooth vector field in $X$. The coupling between the subsystems is such that the first subsystem is independent of the other subsystems. Also, any subsystem $i$, other than the first subsystem, is directly affected by the upstream subsystem $i-1$ but not the downstream subsystem $i+1$. The structure of the  cascade system considered in this article is represented by
\begin{subequations} \label{eqn:decomposition}
    \begin{align}
        x_1^+ &= f_1(x_1,u_1) \\
        x_2^+ &= f_2(x_2,u_2) + g_2(x_1) \\
        \vdots \nonumber \\
        x_N^+ &= f_N(x_N,u_N) + g_N(x_{N-1})
    \end{align}
\end{subequations}
where $f_i$ denotes the local dynamics of state $x_i$, and $g_i$ denotes the dynamics of the coupling between the states $x_i$ and $x_{i-1}$. The subscript $i=1,2, \cdots, N$ represents the $i$th subsystem. The state and input constraints for subsystem $i$ is given by $X_i = \text{proj}_i(X)$ and $U_i = \text{proj}_i(U)$ respectively.

Before we begin our discussion, let us introduce the following definitions which are central and referred to throughout this article.
\begin{definition}[Forward invariant set \cite{blanchini1999}] \label{def:postively_invariant_set}
    A set $R \subseteq X$ is said to be a forward or positively invariant set of the system $x^+=f(x)$ if for every $x \in R$, $f(x) \in R$. 
\end{definition}

\begin{definition}[Control invariant set \cite{blanchini1999}]\label{def:control_invariant_set}
A set $R \subseteq X$ is said to be a control invariant set (CIS) of system (\ref{eqn:system}) if for every $x \in R$, there exist a feedback control law $u = \mu(x) \in U$ such that $R$ is forward invariant for the closed-loop system $f(x,u)$. 
\end{definition}

\begin{definition}[Largest control invariant set \cite{kerrigan2001}]\label{def:largest_control_invariant_set}
A set $R_X \subseteq X$ is said to be the largest (with respect to inclusion) control invariant set of system (\ref{eqn:system}) if $R_X$ is control invariant and contains all other control invariant sets contained in $X$.
\end{definition}

In general, the state constraint $X$ for a given system is not control invariant. However, one may wish to find the largest control invariant set $R_X$ contained in $X$ for controller design and assessment purposes.  The goal of this article is to present a framework for computing an approximation of the largest control invariant set $R_X$ when the state dimension $n$ is too large to make the standard GIS algorithm tractable \cite{decardi2021}. At present, this happens when $n > 4$. This is because the number of cells generated in the standard GIS algorithm grows exponentially in the state dimension.

Our solution to computing an approximation of $R_X$ is to decompose system \eqref{eqn:system} into $M$ lower dimensional subsystems with overlapping states.

\subsection{Graph-based invariant set algorithm}

In this section, we briefly introduce the concept of graph-based invariant set (GIS) algorithm. The reader may refer to the works of Decardi-Nelson and Liu \cite{decardi2021} and Osipenko \cite{osipenko2007} for a more detailed discussion.

Given a dynamical system in the form of Equation~\eqref{eqn:system}, the idea is to cast the goal of finding $R_X$ as a graph theoretical problem. To achieve this, we first rewrite the function $f$ as set a set-valued map parameterized by the input 
\begin{equation} \label{eqn:differential_inclusion}
    F(x) := f(x,U) =  \{f(x,u) \}_{\cup_{u \in U}}
\end{equation}
This makes system~\eqref{eqn:system} a difference inclusion of the form
\begin{equation}\label{eqn:difference_inclusion}
    x^+ \in F(x)
\end{equation}
The function $F$ associates the state $x$ with the subset $F(x)$ of all feasible next states. To construct a digraph representation (also known as symbolic image) of system~\eqref{eqn:system}, the state constraint $X$ first need to be quantized into a collection of cells $\mathcal{C} = \{ B_1, B_2, \cdots, B_l \}$ such that
\begin{equation} \label{eqn:cell_union}
    R = \cup_{B_i \in \mathcal{C}}
\end{equation}
is an outer covering of $R_X$. The diameter $d$ of the collection is given as 
\begin{equation*}
    d = \text{diam}(\mathcal{C}) := \max_{B_i \in \mathcal{C}}\text{diam}(B_i)
\end{equation*} 
where diam($B_i$) = $\text{sup}\{ |x-y|:x,y \in B_i \}$.

Following the quantization of the state constraint, the dynamics of system~\eqref{eqn:system} is approximated using a directed graph $G=(V,E)$ such that
\begin{align}
    V &= \mathcal{C} ~ and \label{eqn:vertices}\\
    E &= \{ (B_i, B_j) \in \mathcal{C} \times \mathcal{C} ~ | ~ F(B_i) \cap B_j \neq \emptyset  \} \label{eqn:edges}
\end{align}
where $V$ is the vertices of the directed graph and $E$ denotes the set of ordered pair of vertices known as edges. From Equations~\eqref{eqn:vertices} and \eqref{eqn:edges}, it is easy to see that the finer the quantization (the smaller the diameter of the collection), the more accurate the approximation of the dynamics of system~\eqref{eqn:system} using a directed graph $G$.

The cells that outer approximate $R_X$ may be obtained by analyzing the resulting directed graph $G$. Specifically, the goal is to split the vertices of $G$ into two groups namely, leaving and non-leaving cells. The leaving cells are cells that have only finite paths passing through them whereas the nonleaving cells are the cells that have infinite admissible paths passing through them. Let us recall the following definitions in graph theory.
\begin{definition}[Strongly connected graph] \label{def:strongly_connected_graph}
    A directed graph $G=(V,E)$ is said to be strongly connected if there is an admissible path in both directions between each pair of vertices of the graph.
\end{definition}

\begin{definition}[Strongly connected component (SCC)] \label{def:strongly_connected_component}
    A strongly connected component of a directed graph $G=(V,E)$ is a maximal strongly connected subgraph of $G$. 
\end{definition}

Algorithmically, finding the non-leaving cells involves finding the SCC of $G$ and any vertex that is not in the SCC but has a path to any vertex in the SCC. A graphical illustration of the shown in Figure~\ref{fig:gis_algorithm}. 
 \begin{figure}[tbp] 
     \center{\includegraphics[width=0.6\columnwidth]{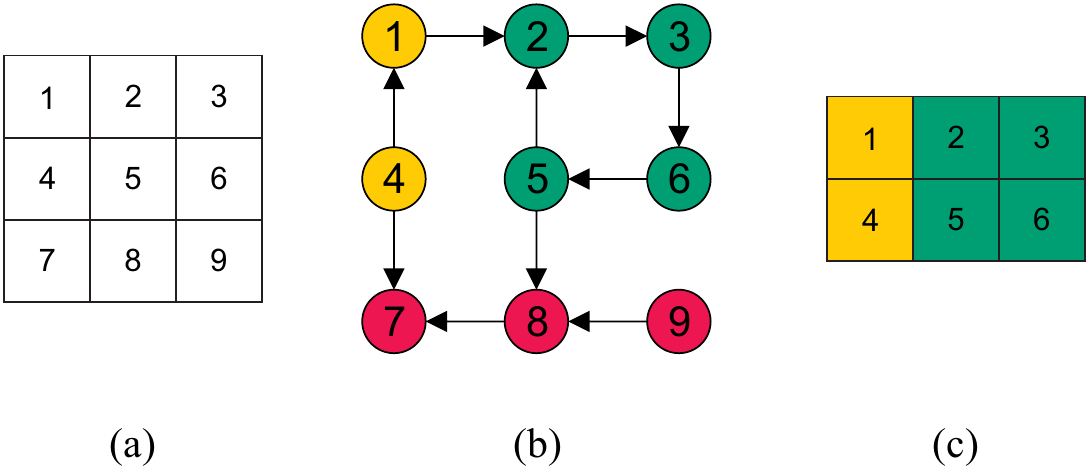}}
     \caption{\label{fig:gis_algorithm} A graphical depiction of the GIS algorithm. (a) Quantization of the state constraint. (b) Graph construction and analysis. The green cells constitute the SCC; the yellow cells have an admissible path to any vertex of the SCC; the red cells do not belong to the SCC nor have a path to the SCC (c) The final selected cells with infinite admissible paths passing through them. The remaining or leaving cells are discarded.}
 \end{figure}
The following proposition summarizes how an outer approximation of $R_X$ may be obtained using graph-based tools.

 \begin{proposition}[Non-leaving cells \cite{decardi2021}] \label{prop:non-leaving_cells}
    Let G=(V,E) be a directed graph approximation of the dynamics of system \eqref{eqn:system} with respect to the collection $\mathcal{C}$. Then
    \begin{enumerate}[i.]
        \item The vertices $V_s$ of the SCC of G has infinite admissible paths passing through them.

        \item any element of $V$ but not in $V_s$, that is $V/V_s$, with a path to at least one path to any element of the SCC also has an infinite admissible path passing through it.

        \item the union of the elements in (i) and (ii), denoted as $I^+(G)$, is a closed neighborhood of the largest control invariant set $R_X$ contained in $X$, that is
        \begin{equation*}
            R_X \subseteq I^+(G)
        \end{equation*}
    
    \end{enumerate}
 \end{proposition}

Furthermore, it can be shown that $R$ converges to the largest control invariant set $R_X$ as $d \rightarrow 0$ \cite{decardi2021}.

\section{System decomposition and set invariance} \label{sec:system_decomposition}

In trying to alleviate the exponential cell growth in the GIS algorithm, we propose to decompose the overall system into smaller subsystems. This makes it computationally tractable for control invariant set approximation using the GIS algorithm. System decomposition is a typical way of addressing the computational challenges associated with control and state estimation of large scale dynamical systems. This is also the case for control invariant set calculation. The goal is to divide the original system into many small subsystems thus making the computations tractable. Usually, the decomposition is achieved by dividing the system in such a way that the subsystems have weak to no coupling or interconnection. 

In this section, we investigate the system structures suitable for decomposition and control invariant set computation. In particular, we consider simple parallel and series/cascade system structures. Throughout this section, we restrict the discussion to two and three dimensional system structures namely, series and parallel system structures to make the ideas presented here easier to follow. The ideas presented here can easily be generalized to much higher dimensional systems. 

We begin this section by analyzing two types of system structures, their decomposition and the ability to reconstruct the solution from the solution of the subsystems. Thereafter we described the method of overlapping decomposition, which is a precursor for the proposed distributed algorithm.

\subsection{System structures and invariance}

Consider the disjoint system structure 

\begin{equation}\label{eqn:parallel_system}
         x^+ = \hat{f}(x) = 
         \begin{bmatrix}
         \hat{f}_1(x_1)  \\
         \hat{f}_2(x_2)  
         \end{bmatrix} 
         =
         \begin{bmatrix}
         A_{11} & 0 \\
         0 & A_{22}
         \end{bmatrix} 
         x
\end{equation}
where $x = [x_1~x_2]^T$ subject to the state constraint $X$. A graphical depiction of the system is shown in Figure~\ref{fig:parallel_system_structure}. 
\begin{figure}[tbp] 
    \center{\includegraphics[width=0.08\columnwidth]{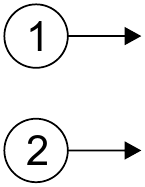}}
    \caption{\label{fig:parallel_system_structure} Parallel or independent system structure}
\end{figure}
Since the evolution of the system states are independent of each other, the system equations can be intuitively decomposed into the following subsystems
\begin{equation}\label{eqn:parallel_decomposition}
    S_1: x_1^+ = \hat{f}_1(x_1) = A_{11} x_1 \text{ and } S_2: x_2^+ = \hat{f}_2(x_2) = A_{22} x_2 
\end{equation}
with subspaces $X_1$ = proj$_1$($X$) and $X_2$ = proj$_2$($X$) respectively. To construct a graph representation of the dynamics of each subsystem, the subspaces $X_1$ and $X_2$ need to be quantized into $\mathcal{C}_1$ and $\mathcal{C}_2$. In what follows, we show how the control invariant set can be computed for the full system from the subsystem information. 
\begin{definition}[Graph Cartesian product \cite{TROTTER1978}] \label{defn:graph_cartesian_product}
    The Cartesian product of the directed graphs $G_1 = (V_1,E_1)$ and $G_2 = (V_2,E_2)$, denoted $G_1 \times G_2$, is a graph $G=(V,E)$ such that $V = V_1 \times V_2$ and for any two points $u=(u_1,u_2)$ and $v=(v_1,v_2)$ in $V$, the directed edge $(u,v) \in E$ whenever $u_1 = v_1$ and $(u_2,v_2) \in E_2$, or $u_2 = v_2$ and $(u_1,v_1) \in E_1$.
\end{definition}
An illustration of the Cartesian product of two fictitious digraphs $G_1$ and $G_2$ is presented in Figure~\ref{fig:graph_cartesian_product}.

\begin{figure}[tbp] 
    \center{\includegraphics[width=0.4\columnwidth]{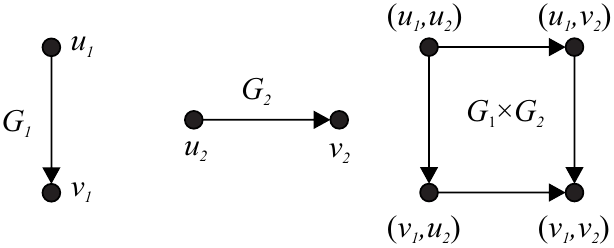}}
    \caption{\label{fig:graph_cartesian_product} The Cartesian product of two graphs}
\end{figure}

\begin{assumption}[Subsystem quantization] \label{asp:constraint_quantization}
    
The set $\mathcal{C}$, which is a quantization of $X$, is the cross product of $\mathcal{C}_1$ and $\mathcal{C}_2$, that is $\mathcal{C} = \mathcal{C}_1 \times \mathcal{C}_2$.
\end{assumption}
\begin{proposition}\label{prop:parallel_cartesian_graph_prop}

Consider the system described by Equation~\eqref{eqn:parallel_system} and decomposed in the form of Equation~\eqref{eqn:parallel_decomposition}. Let $G_1=(V_1,E_1)$ and $G_2=(V_2,E_2)$ be a directed graph representation of the dynamics of each subsystem $S_1$ and $S_2$ respectively based on the sets $\mathcal{C}_1$ and $\mathcal{C}_2$. Also, let $G=(V,E)$ be the directed graph representation of system~\eqref{eqn:parallel_system} based on $\mathcal{C}$. If Assumption~\ref{asp:constraint_quantization} holds, then the full system solution can be exactly obtained from the subsystem solutions, that is,
\begin{equation*}
    I^+(G) = I^+(G_1 \times G_2)
\end{equation*}

\end{proposition}

\begin{proof}

To prove the above assertion, we first need to show that the vertices of the two graphs $G$ and $G_1 \times G_2$ are equal. Thereafter, we need to show that if there is an admissible path between any two cells in $G$, then there is an equivalent admissible path between those same two cells in $G_1 \times G_2$. 

First, we show that the vertices of the two graphs $G$ and $G_1 \times G_2$ are equal. It follows directly from Assumption~\ref{asp:constraint_quantization} that $\mathcal{C} = \mathcal{C}_1 \times \mathcal{C}_2$. Therefore from Definition~\ref{defn:graph_cartesian_product}, we have that $V = V_1 \times V_2$. 

Now, we show by that if there is a path between any two cells in $G$, then there exist a path between those same two cells in $G_1 \times G_2$. Let $B_i=(u_1,u_2)$ and $B_j=(v_1,v_2)$ be any two cells in $V$ such that $u_1$ and $v_1$ are in $V_1$, and $u_2$ and $v_2$ are cells in $V_2$.
From the graph construction procedure in Equation~\eqref{eqn:edges}, we have that
\begin{align*}
    E &= \{ (B_i,B_j) \in \mathcal{C} \times \mathcal{C} : \hat{f}(B_i) \cap B_j \neq \emptyset \} \\
    &= \{ (B_i,B_j)  \in \mathcal{C} \times \mathcal{C} : \hat{f}_1(u_1) \cap v_1 \neq \emptyset \land \hat{f}_2(u_2) \cap v_2 \neq \emptyset\} \\
    &= \{ (B_i,B_j)  \in \mathcal{C} \times \mathcal{C} : (u_1,v_1) \in E_1 \land (u_2,v_2) \in E_2 \}
\end{align*}
From the definition of $G_1 \times G_2$ and $E$, we have that the path $(u_1,u_2) \rightarrow (u_1,v_2) \rightarrow (v_1,v_2)$ and $(u_1,u_2) \rightarrow (v_1,u_2) \rightarrow (v_1,v_2)$ exist in $G_1 \times G_2$. This implies that there is an admissible path from $(u_1,u_2)$ to $(v_1,v_2)$ in $G_1 \times G_2$. It then follows that if an infinite admissible path passes through the cell $B_i$ in $G$, then an infinite admissible path also passes through the cell $B_i=(u_1,u_2)$ in $G_1 \times G_2$. Hence we have that
\begin{equation*}
    I^+(G) = I^+(G_1 \times G_2),
\end{equation*}
which completes the proof.
\end{proof}
While we have shown that an approximation of $R_X$ can be obtained for the disjoint system by constructing the digraph $G_1 \times G_2$, in most cases this will require a considerable amount of memory to store the graph for higher dimensional systems. The following theorem shows how an approximation of the largest control invariant set can be exactly obtained from the solutions of the subsystems.

\begin{theorem}\label{thm:parallel_reconstruction_thm}

Consider the system described by Equation~\eqref{eqn:parallel_system} and decomposed in the form of Equation~\eqref{eqn:parallel_decomposition}. Let $G_1=(V_1,E_1)$ and $G_2=(V_2,E_2)$ be a directed graph representation of the dynamics of each subsystem $S_1$ and $S_2$ respectively based on the sets $\mathcal{C}_1$ and $\mathcal{C}_2$. Also, let $G=(V,E)$ be the directed graph representation of system~\eqref{eqn:parallel_system} based on $\mathcal{C}$. If Assumption~\ref{asp:constraint_quantization} holds, then the full system solution can be exactly obtained from the subsystem solutions, that is,
\begin{equation*}
    I^+(G) = I^+(G_1) \times I^+(G_2)
\end{equation*}
\end{theorem}

\begin{proof}

We know from Proposition~\ref{prop:parallel_cartesian_graph_prop} that $I^+(G) = I^+(G_1 \times G_2)$. Let $B_i=(u_1,u_2) \in V$ where $u_1 \in V_1$ and $u_2 \in V_2$. It can be inferred from the definition of $G_1 \times G_2$ that wherever there is an infinite admissible path passing through $B_i$, there is also infinite admissible paths passing through $u_1$ in $G_1$ and $u_2$ in $G_2$. This implies that
\begin{equation*}
    I^+(G) = I^+(G_1) \times I^+(G_2) \qedhere
\end{equation*}
\end{proof}
Theorem~\ref{thm:parallel_reconstruction_thm} shows that the graph representation of the subsystems can be analyzed independently. Then the cells that approximate the largest control invariant set for full system can be computed by finding the Cartesian product of the individual subsystem solutions. This is computationally more efficient since a large graph is not constructed for analysis. To illustrate the ideas presented so far, let us consider Example~\ref{exp:disjoint_example}.

\begin{example} \label{exp:disjoint_example}
Consider the system
\begin{equation}\label{eqn:parallel_projection_example}
         x^+ =
         \begin{bmatrix}
         2 & 0 \\ 
         0 & 2 
         \end{bmatrix} 
         x + 
         \begin{bmatrix} 1 & 0 \\ 0 & 1\end{bmatrix} u
\end{equation}
where the constraints on the states and input are $X = \{ x \in \mathbb{R}^2 : \| x \|_{\infty} \leq 5 \} $ and $U = \{ u \in \mathbb{R}^2 : \| u \|_{\infty} \leq 1 \}$ respectively. 
\end{example}

\begin{figure}[tbp] 
    \center{\includegraphics[width=0.9\columnwidth]{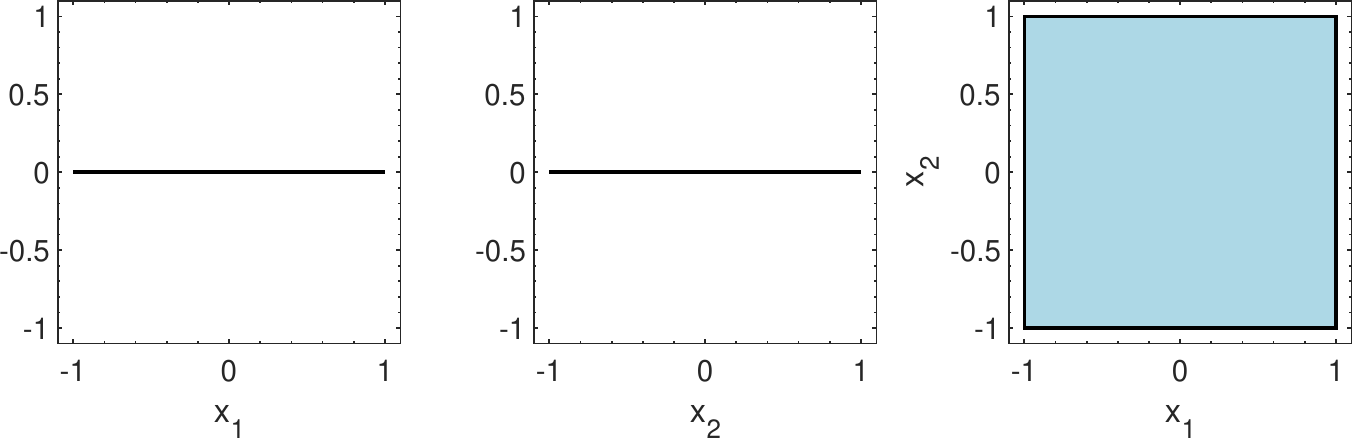}}
    \caption{\label{fig:disjoint_example} Decomposition and full solution reconstruction for system~(\ref{eqn:parallel_projection_example}). Left: Solution of Subsystem~1. Middle: Solution of Subsystem~2. Right: The reconstructed full system solution from the subsystem solutions (black line) and the actual solution of the full system (light blue region).}
\end{figure}

The system in Example~\ref{exp:disjoint_example} can be decomposed into two subsystems in a form similar to Equation~\eqref{eqn:parallel_decomposition}. Figure~\ref{fig:disjoint_example} shows the largest control invariant sets $R_{X_1}$ and $R_{X_2}$ for each of the subsystems, the largest control invariant set $R_X$ for the full system and the reconstructed solution from the control invariant sets of the subsystems. It can be seen that the full system solution can be exactly obtained from the subsystem solution without incurring any error due to the decomposition. This agrees well with our earlier arguments in Theorem~\ref{thm:parallel_reconstruction_thm}, that $R_X = R_{X_1} \times R_{X_2}$. Ultimately, this shows that we can freely move between the control invariant sets of the lower dimension subsystems and the control invariant set for the full system without incurring an losses. Geometrically, the structure of the control invariant set is always going to be a box.

While an approximation of the largest control invariant set can be obtained from the solutions of the subsystems without incurring any errors due to decomposition for disjoint subsystems, such systems rarely occur in the real world. In most cases, the subsystems may be interconnected. Let us consider one of such cases where is coupling between the system states. Consider the system
\begin{equation}\label{eqn:series_structure}
         x^+ = \hat{f}(x)
         =
         \begin{bmatrix}
         A_{11} & 0 \\ 
         A_{21} & A_{22} 
         \end{bmatrix} 
         x
\end{equation}
where $x = [x_1~x_2]^T$ subject to the state constraint $X$. A graphical depiction of the system is shown in Figure~\ref{fig:series_system_structure}. Similar to the disjoint structure in Equation~\eqref{eqn:parallel_system}, system~\eqref{eqn:series_structure} can be decomposed into two subsystems such that
\begin{equation}\label{eqn:series_decomposition}
    S_1: x_1^+ = \hat{f}_1(x_1) = A_{11} x_1 \text{ and } S_2: x_2^+ = \hat{f}_2(x_2) + \hat{g}_2(\tilde{x}_1) = A_{22} x_2 + A_{21}\tilde{x}_1 
\end{equation}
with subspaces $X_1$ = proj$_1$($X$) and $X_2$ = proj$_2$($X$) respectively. Notice that while Subsystem~1 is independent of the dynamics of Subsystem~2, Subsystem~2 is dependent on the dynamics of Subsystem~1. The state information about Subsystem~1 required by Subsystem~2 is denoted by $\tilde{x}_1$ in Equation~\eqref{eqn:series_decomposition}. To work with such a decomposition, the fundamental issues that needs to be addressed in this scenario are how to obtain the values of $\tilde{x}_1$ and how to use this information in Subsystem~2.
\begin{figure}[tbp] 
    \center{\includegraphics[width=0.2\columnwidth]{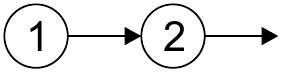}}
    \caption{\label{fig:series_system_structure} Series or connected system structures}
\end{figure}

It is easy to see that the range of values for $\tilde{x}_1$ can be obtained from the solution of Subsystem~1. For the choice of usage of the $\tilde{x}_1$ information, one approach is to use the $\tilde{x}_1$ information as an input in Subsystem~2. However, this more often than not result in an over approximation of the solution of the full system.

\begin{theorem}\label{thm:series_reconstruction_thm}

Consider the system described by Equation~\eqref{eqn:series_structure} and decomposed in the form of Equation~\eqref{eqn:series_decomposition}. Let $G_1=(V_1,E_1)$ and $G_2=(V_2,E_2)$ be a directed graph representation of the dynamics of each subsystem $S_1$ and $S_2$ respectively based on the sets $\mathcal{C}_1$ and $\mathcal{C}_2$. Also, let $G=(V,E)$ be the directed graph representation of system~\eqref{eqn:series_structure} based on $\mathcal{C}$. If Assumption~\ref{asp:constraint_quantization} holds, then 

\begin{equation*}
    I^+(G) \subseteq I^+(G_1) \times I^+(G_2)
\end{equation*}

\end{theorem}
\begin{proof}

 This proof follows along the same lines as Theorem~\ref{thm:parallel_reconstruction_thm}. We know from Proposition~\ref{prop:parallel_cartesian_graph_prop} that $I^+(G) = I^+(G_1 \times G_2)$. Let $B_i=(u_1,u_2) \in V$ where $u_1 \in V_1$ and $u_2 \in V_2$. It can be inferred from the definition of $G_1 \times G_2$ that wherever there is an infinite admissible path passing through $B_i$, there is also infinite admissible paths passing through $u_1$ in $G_1$ and $u_2$ in $G_2$. This implies that
\begin{equation*}
    I^+(G) \subseteq I^+(G_1) \times I^+(G_2) \qedhere
\end{equation*}

\end{proof}
The reverse of Theorem~\ref{thm:series_reconstruction_thm}, that is $I^+(G_1) \times I^+(G_2) \subseteq I^+(G)$ is necessarily not true. To demonstrate this, let us consider the following example.

\begin{example}\label{exp:connected_example}
Consider the system
 \begin{equation}\label{eqn:series_projection_example}
          x^+ =
          \begin{bmatrix}
          2 & 0 \\ 
          1 & 2 
          \end{bmatrix} 
          x + 
          \begin{bmatrix} 1 & 0 \\ 0 & 1\end{bmatrix} u
 \end{equation}
where the constraints on the states and input are $X = \{ x \in \mathbb{R}^2 : \| x \|_{\infty} \leq 5 \} $ and $U = \{ u \in \mathbb{R}^2 : \| u \|_{\infty} \leq 1 \}$ respectively. 
\end{example} 

As an illustration of our arguments so far, consider the system in Example~\ref{exp:connected_example}. System~\eqref{eqn:series_projection_example} can be decomposed into two subsystems following Equation~\eqref{eqn:series_decomposition}. Figure~\ref{fig:connected_example} shows the solution of the subsystems $R_{X_1}$ and $R_{X_2}$, when $\tilde{x}_1$ is treated as an input in Subsystem~2, the reconstructed solution of the full system from the solution of the subsystems, and the solution of the full system $R_X$. It can be seen that $R_{X} \subseteq R_{X_1} \times R_{X_2}$. Hence, there exists some states in $R_{X_1} \times R_{X_2}$ which are not in $R_X$. 

\begin{figure}[tbp] 
    \center{\includegraphics[width=0.9\columnwidth]{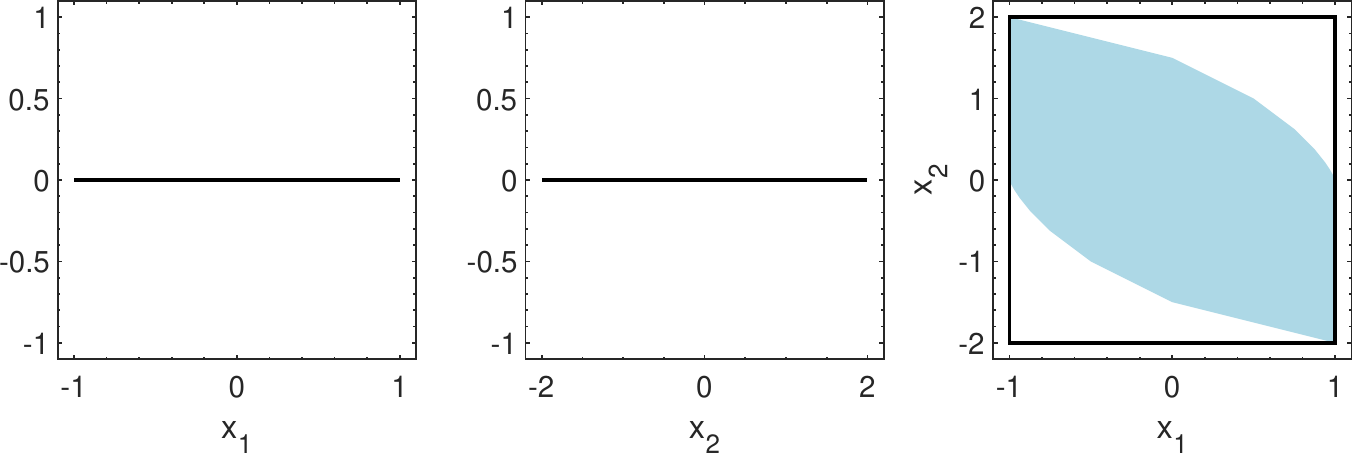}}
    \caption{\label{fig:connected_example} Decomposition and full solution reconstruction for system~(\ref{eqn:series_projection_example}). Left: Solution of Subsystem~1. Middle: Solution of Subsystem~2 when $\tilde{x}_1$ is treated as an input. Right: The reconstructed full system solution from the subsystem solutions (black line) and the actual solution of the full system (light blue region).}
\end{figure}

To address the issue of over approximation, several studies have proposed to treat $\tilde{x}_1$ as an uncertain input in Subsystem~2 \cite{li2020,KAYNAMA2013}. Let us introduce the following definition
\begin{definition}[Robust control invariant set]
A set $R$ is said to be a robust control invariant set of the system $x^+ = f(x,u,w)$, if for every $x \in R$, there exist a feedback control law $u = \mu(x) \in U$ such that $R$ is forward invariant for the closed-loop system $f(x,u,w)$ for all uncertainty $w$ in the uncertainty set $W$.
\end{definition}

\noindent By treating $\tilde{x}_1$ as an uncertain input, a worse case scenario analysis can be considered during the computation of the solution of Subsystem~2. This results in a robust control invariant set solution for Subsystem~2. While this approach works, there are several downsides to it. First, it is neither trivial nor easy to compute robust control invariant sets. In fact, more resources are needed to compute robust control invariant sets than to compute control invariant sets. Second, the resulting reconstructed set is often conservative or small compared to the control invariant set for the full system. In some instances, the computation may end up with an empty set for Subsystem~2 which is undesired. Let us consider the system in Example~\ref{exp:connected_example} again. When $\tilde{x}_1$ is treated as an uncertainty in its worst case, $R_{X_2} = \emptyset$, implying that $R_{X_2} \times R_{X_2} = \emptyset$, which is not desirable. This also shows that $R_{X_1} \times R_{X_2} \subseteq R_{X}$ when the $\tilde{x}_1$ is treated as an uncertainty, which is a conservative result. 

In a nutshell, it can be seen that when the subsystems are connected, the full system solution cannot be easily obtained from the solution of the subsystems. Further scrutinizing the actual control invariant set for system~\eqref{eqn:series_projection_example} in Figure~\ref{fig:connected_example}, it can be seen that $\tilde{x}_1$ actually behaves as an input at some points in $x_2$ and as an uncertainty in other $x_2$ locations. This is the central idea we upon which we develop our distributed algorithm.

\subsection{Overlapping system decomposition}


In the preceding section, we have discussed the implication of decomposition of set invariance using 2 dimensional systems. But what how do we decompose connected systems when the system dimension is greater than 2? In this section, we briefly discuss how to decompose three dimensional cascade systems.

Let us consider the following system
\begin{equation}\label{eqn:series_structure_3d}
         x^+ = \hat{f}(x)
         =
         \begin{bmatrix}
         A_{11} & 0 & 0 \\
         A_{21} & A_{22} & 0 \\
         0 & A_{32} & A_{33} \\
         \end{bmatrix} 
         x
\end{equation}
where $x = [x_1~x_2~x_3]^T$ subject to the state constraint $X$. 
\begin{figure}[tbp] 
    \center{\includegraphics[width=0.8\columnwidth]{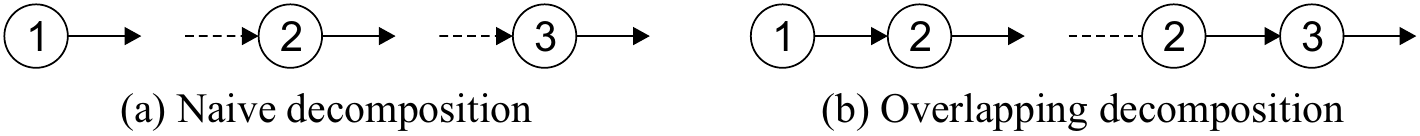}}
    \caption{\label{fig:series_decomposition} Different decomposition strategies for the cascade system}
\end{figure}
There are different ways to decompose system~(\ref{eqn:series_structure_3d}) into several subsystems. A simple but naive decomposition is to consider each state in Equation~(\ref{eqn:series_structure_3d}) separately as shown in Figure~\ref{fig:series_decomposition}. 
This way, three subsystems will be obtained with each subsystem focusing on a single state. While this look simple and is similar to the parallel decomposition, it require information about the state in the preceding subsystem. Moreover, the control invariant sets obtained this way based on the subsystems are not very useful in reconstructing the control invariant set of the original system since they do not contain the interconnection information about the neighboring states. The only way to reconstruct the control invariant set from the subsystem solutions is to find the Cartesian product of the resulting control invariant set for the each subsystem. This will ultimately result in a large over approximation of $R_X$.

A more useful way is to decompose system~\eqref{eqn:series_structure_3d} into subsystems such that each subsystem is chained to its neighboring subsystem as shown in Figure~\ref{fig:series_decomposition}. This way, the interactions between the states can be accounted for. As an example, system~(\ref{eqn:series_structure_3d}) can be decomposed into two subsystems that share a common part. This can be considered as equivalent to expanding the system to a higher dimension since the overlapping states are repeated in the state equation as shown in Equation~(\ref{eqn:expanded_system}). The overlapping states play an important role in reconstructing the control invariant set from the distributed computing results. The role of the overlapping states will be made clear in the later discussion. 

More formally, let $z \in \m{R}^{n_1 + n_2}$ denote the states for the expanded system for system~(\ref{eqn:series_structure}) such that $z = [x_1~x_2~x_2~x_3]^T$ is the expanded system state vector with $z_1 = [x_1~x_2]^T \in \m{R}^{n_1}$ and $z_2 = [x_2~x_3]^T  \in \m{R}^{n_2}$ being the subsystem states. The two subsystems are shown in Equation~(\ref{eqn:series_subsystems}). The corresponding decomposition is depicted in Figure~\ref{fig:series_decomposition}. 

\begin{equation} \label{eqn:expanded_system}
     z^+
     =
     \left[
     \begin{array}{cc|cc}
     A_{11} & 0 & 0 & 0 \\
     A_{21} & A_{22} & 0 & 0 \\ \hline
     A_{21} & 0 & A_{22} & 0 \\
     0 & 0 & A_{32} & A_{33} \\
     \end{array} \right]
     z
 \end{equation} 

  \begin{equation} \label{eqn:series_subsystems}
      S_1: z_1^+ = f_1(z_1)
      =
      \begin{bmatrix}
      A_{11} & 0 \\
      A_{21} & A_{22} 
      \end{bmatrix} 
      z_1 \quad 
      S_2: z_2^+ = f_2(z_2,\tilde{x}_1)
      =
      \begin{bmatrix}
      A_{22} & 0 \\
      A_{32} & A_{33} 
      \end{bmatrix} 
      z_2 
      +
      \begin{bmatrix}
      A_{21}  \\
      0 
      \end{bmatrix} 
      \tilde{x}_1 
  \end{equation}

 Notice that while Subsystem 1 is independent of any external state, Subsystem 2 requires information about $x_1$ similar to our earlier discussions. We will refer to the information about $x_1$ required by Subsystem~2 as missing state information, denoted by $\tilde{x}_1$, in the next section. How the missing state information is obtained will be discussed in the subsequent section. As can be seen in the above equations, the decomposed system has a higher dimension than the original system i.e. $n_1 + n_2 > n$. While the above example used a linear system, this type of decomposition can be readily applied to nonlinear systems with exogenous inputs and serves as a precursor to the proposed algorithm. Unfortunately, there is no general way of reconstructing the control invariant set for the overall system from the solution of the subsystems. This is because of the coupling between the two subsystems and that in Subsystem 2, though $x_1$ information is used, the dynamics is not considered. 

\section{Computation of the largest control invariant set via system decomposition} \label{sec:cis_computation}

In this section, we present an algorithm to compute an approximation of $R_X$ via system decomposition. We focus on the cascade structure presented in Section~\ref{sec:system_decomposition}. Given a cascade system with its associated overlapping decomposition, the algorithm computes the control invariant set for each of the subsystems, first in a decentralized manner, followed by a distributed. Thereafter, the control invariant for the overall system is reconstructed and subset of the cells tested for control invariance. The central idea employed here is that the missing state information is treated as both a regular input and an uncertain input. This way, the two aspects of the missing state is utilized. As we will show in this section, the use of graph analysis makes it easier for this to be achieved. Throughout this section, we assume there are only two subsystems $S_1$ and $S_2$, that is $M=2$, as described in Equations~\eqref{eqn:series_subsystems} with subspaces $X_1 = \text{proj}_1 (X)$ and $X_2 = \text{proj}_2 (X)$ respectively. 

We begin this section by first describing procedure for the decentralized and distributed computation. Thereafter, we present the procedure for the reconstruction and validation of the solution. We end this section by conducting a brief analyzes of the complexity of the proposed algorithm. 

\subsection{Decentralized and distributed computation}

The decentralized and the distributed computation makes use of the standard GIS algorithm to determine the cells which approximate the largest control invariant sets ($R_{X_1}$ and $R_{X_2}$) for each of the subsystems. We denote the solutions of the subsystems by $R_1$ and $R_2$ respectively. During the computations, the directed graphs $G_1$ and $G_2$ for each subsystem are constructed based on the collections $\mathcal{C}_{1}$ and $\mathcal{C}_{2}$ of each subsystem respectively. Further analysis on these graphs is vital for the reconstruction of the solution of the overall system.

In the decentralized step, each subsystem is treated as independent of each other. This allows for the computation of $R_1$ and $R_2$ in parallel. The goal is to quickly obtain an approximation of the largest control invariant sets and reduce the search space for further computations. To make this possible, the interval for missing state information for $x_1$ in Subsystem~2 is held constant during the computation. In this case, an approximation which is obtained by projecting the state constraint onto the dimension of $x_1$ may be used, that is
\begin{equation}
    \tilde{x}_1 = \text{proj}_{u} X
\end{equation}
where $u$ denotes the dimension of $x_1$. Once the computation begin, there is no communication between the solutions of the two subsystems. The algorithm for the decentralized step is summarized in Algorithm~\ref{alg:decentralized_computation}. The algorithm takes as input the subsystem equations, the quantized subspaces for each subsystem, the input constraint and the number of subsystems. As mentioned earlier, this is fixed since only two subsystems are considered here. The algorithm returns estimates of the control invariant sets for each of the subsystems as well as their associated digraphs.

\begin{algorithm}

\caption{Decentralized computation} \label{alg:decentralized_computation}

\KwIn{Subsystems in Equation~\eqref{eqn:series_subsystems} and the collections $\mathcal{C}_1$ and $\mathcal{C}_2$, $U$, $M=2$}
\KwOut{$G_i$, $R_i$, $i=1, \cdots, M$   }
$G_1 \leftarrow$ Construct the directed graph representation of $S_1$ \\
$R_1 \leftarrow$ $I^+(G_1)$\\
\For{$i=2 \cdots M$ } { \tcp{Can be done in parallel}
    $\tilde{x}_{i-1} \leftarrow \text{proj}_{i-1}X$ \\
    $G_i \leftarrow$ Construct the directed graph representation of $S_i$ utilizing $\tilde{x}_{i-1}$ \\
    $R_i \leftarrow$ $I^+(G_i)$ \\
}

\textbf{return $R_{i},G_i, (i=1,\cdots,M)$}

\end{algorithm}

In the distributed step, the control invariant set for each subsystem is computed sequentially. Thus, the solution for the immediately preceding subsystem is used to estimate the missing state information for the current subsystem.  While the missing state information for each subsystem is obtained from the state constraint in the decentralized computation, the missing state information is dynamically obtained from the neighboring subsystem solution in the distributed computation. This is the key difference between the decentralized computation and the distributed computation. It is also the reason for computing the solution for each subsystem in a sequential fashion. The procedure for dynamically estimating the missing state information in the computation of the control invariant set of Subsystem~2 is presented in Figure~\ref{fig:missing_state_estimation}. 
\begin{figure}[tbp] 
    \center{\includegraphics[width=0.5\columnwidth]{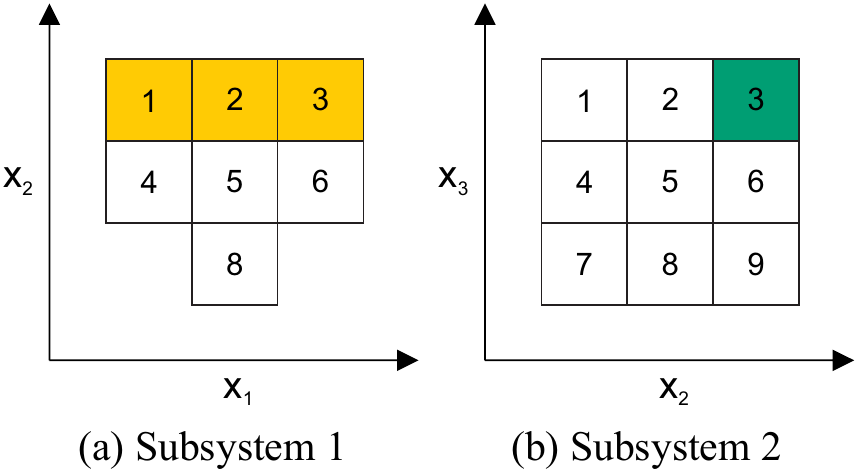}}
    \caption{\label{fig:missing_state_estimation} Procedure for dynamically estimating the missing $x_1$ information for each cell in Subsystem 2. (a) Solution of Subsystem 1 $R_{1}$. The yellow cells indicate the cells that correspond to $B_3$ in the solution of Subsystem 2. Merging and projecting these cells onto the $x_1$ dimension produces the range of $x_1$ for $B_3$. (b) Solution of Subsystem 2 $R_{2}$. The green cell indicate the cell whose missing state information is being estimated. This procedure is repeated for all other cells in $R_{2}$.}
\end{figure}
By allowing unidirectional communication through dynamic estimation of the missing state information, the solution of each subsystem is further refined. Furthermore, a more realistic graph representation of the dynamics of each subsystem is obtained. This will be useful during the set reconstruction step of the algorithm. The algorithm for the distributed computation is presented in Algorithm~\ref{alg:distributed_computation}. It is similar to the algorithm used for the decentralized computation. As mentioned earlier, the key difference is how the missing state information $\tilde{x}_i$ is obtained.

\begin{algorithm}
\caption{Distributed computation} \label{alg:distributed_computation}
\KwIn{Subsystems in Equation~\eqref{eqn:series_subsystems} and the collections $\mathcal{C}_1$ and $\mathcal{C}_2$, $U$, $M=2$}
\KwOut{$G_i$, $R_i$, $i=1, \cdots, M$   }
$G_1 \leftarrow$ Construct the directed graph representation of $S_1$ \\
$R_1 \leftarrow$ $I^+(G_1)$\\
\For{$i=2 \cdots M$}{
    Estimate the range of $\tilde{x}_{i-1}$ for $S_i$ from $R_{i-1}$ \\ 
    $G_i \leftarrow$ Construct the directed graph representation of $S_i$ utilizing $\tilde{x}_{i-1}$ \\
    $R_i \leftarrow$ $I^+(G_i)$ \\
}

\textbf{return $R_{i},G_i, (i=1,\cdots,M)$}

\end{algorithm}

\subsection{Set reconstruction and validation}

Following the distributed computation step, the solutions from the two subsystems, that is $R_1$ and $R_2$, together with the directed graphs $G_1$ and $G_2$ are used to reconstruct the solution for the overall system $R$. Notice that the Cartesian product of the two sets will result in a 4 dimensional object which is not accurate. This is because of the overlapping decomposition. The procedure for reconstructing $R$ therefore involves projecting the cells in $R_1$ onto the $x_1$ dimension in a dynamic fashion similar to the procedure for estimating the missing state information in the distributed computation. However, the range of $x_1$ is not merged in this case since the focus is reconstruction of the set for the overall system. The procedure for reconstructing the set is described in Figure~\ref{fig:set_reconstruction}.
\begin{figure}[tbp] 
    \center{\includegraphics[width=0.6\columnwidth]{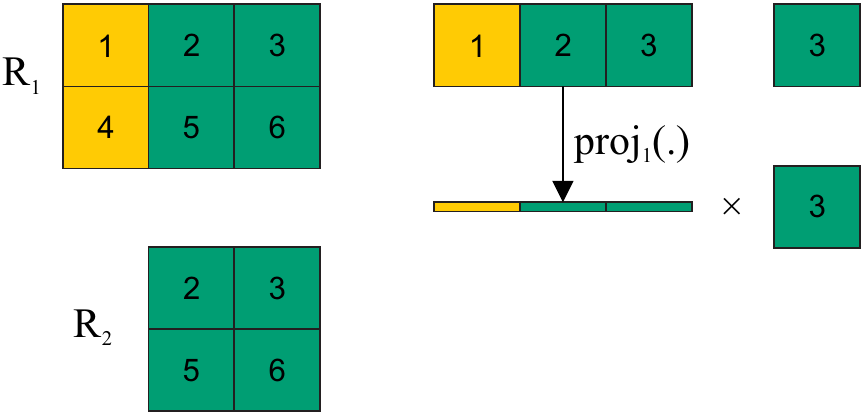}}
    \caption{\label{fig:set_reconstruction} Procedure for reconstructing the control invariant set for the overall system from the subsystem solutions $R_1$ and $R_2$. The procedure is indicated for Cell $B_3$ in $R_2$. Because of the overlap, $B_6$ in $R_2$ has a connection with with $B_1$, $B_2$ and $B_3$ in $R_1$. These cells are then projected onto the $x_1$ dimension to obtain the range of $x_1$ for each cell. The Cartesian product of the $B_3$ with the range of $x_1$ produces 3 dimensional cells (in this case 3 cells). This is repeated for all other cells in $R_2$ to obtain the solution $R$ for the overall system.}
\end{figure}

As described earlier, reconstructing the solution of the overall system from the subsystem solutions may often than not result an approximation error due to the decomposition and coupling between the subsystems. However, compared to the naive decomposition, the approximation error for the overlapping decomposition is expected to be smaller. Nonetheless, the approximation error need to be addressed to ensure that the reconstructed solution from the solution of the subsystems converges to solution when the full system model is used. 

To address the approximation error due to the decomposition, we propose that a subset of the cells forming the reconstructed set needs to be validated for inclusion in the final reconstructed solution. The validation test involves computing the next feasible states of each cell to be tested using the overall system model. Thereafter, the next feasible states of the cell is checked for intersection with the reconstructed set. A cell $B_i$ fails the validation test if its successive states does not intersect the reconstructed set $R$, that is
\begin{equation}
    F(B_i) \cap R = \emptyset
\end{equation}
This is a direct intuition from the definition of control invariance. Any cell that fails the validation test is removed from the reconstructed set.

What now remains is how to select the cells from the reconstructed set $R$ for the validation test. In principle, each cell in the reconstructed set need to be tested. However, this may result in having to check for a large number of cells especially when the system dimension is high. An alternative approach is to further investigate the graph $G_2$ from the distributed computation step to determine the cells that need to be tested. This way, the cells to test will be significantly reduced. The general idea we use to further investigate the graph $G_2$ is to consider the adversarial aspects of the missing state information used to construct $G_2$. We know by definition that the set $R_2$ is made of cells with infinite admissible paths passing through it, that is non-leaving cells. Thus, if there is a direct path from a cell in $R_2$ to a cell not in $R_2$, then that path could have been caused by the missing state $\tilde{x}_1$. By doing this, we consider the contribution of the missing state $\tilde{x}_1$ as a regular input and as an uncertain input.

Given the collection $\mathcal{C}_2$ for Subsystem~2, let $\mathcal{C}^{lc}_2$ denote the leaving cells. It is easy to see that $\mathcal{C}_2 = R_2 \cup \mathcal{C}^{lc}_2$. 
\begin{definition}[Incoming neighbors]
Given the directed graph $G=(V,E)$, the incoming or in neighbors of the vertex $u \in V$ are the nodes $v \in V$ such that the edges $(v,u) \in E$ exist. 
\end{definition}

\noindent The central idea is to find the incoming neighbors of the leaving cells in the graph of Subsystem~2. Then we find those incoming neighbors that have an intersection with the non-leaving cells. Let $\mathcal{C}^{t}_2$ denote the cells that need to be tested for Subsystem~2 and in$(G_2,\mathcal{C}^{lc}_2)$ be the incoming neighbors of the leaving cells on $G_2$. Then the following relationship describes the cells that need to be flagged for further testing after the reconstruction
\begin{equation}\label{eqn:cell_testing}
    \mathcal{C}^{t}_2 = \text{in}(G_2, \mathcal{C}^{lc}_2) \cap R_2
\end{equation}
By obtaining the cells to be tested according to Equation~\eqref{eqn:cell_testing}, we ensure that any cell that has a direct path to any of the leaving cells are verified. The cells to be tested, that is $\mathcal{C}^t$ for the full dimensional problem can be obtained from $\mathcal{C}^t_2$ by following the set reconstruction procedure described earlier for $R_2$.

\begin{figure}[tbp] 
    \center{\includegraphics[width=0.6\columnwidth]{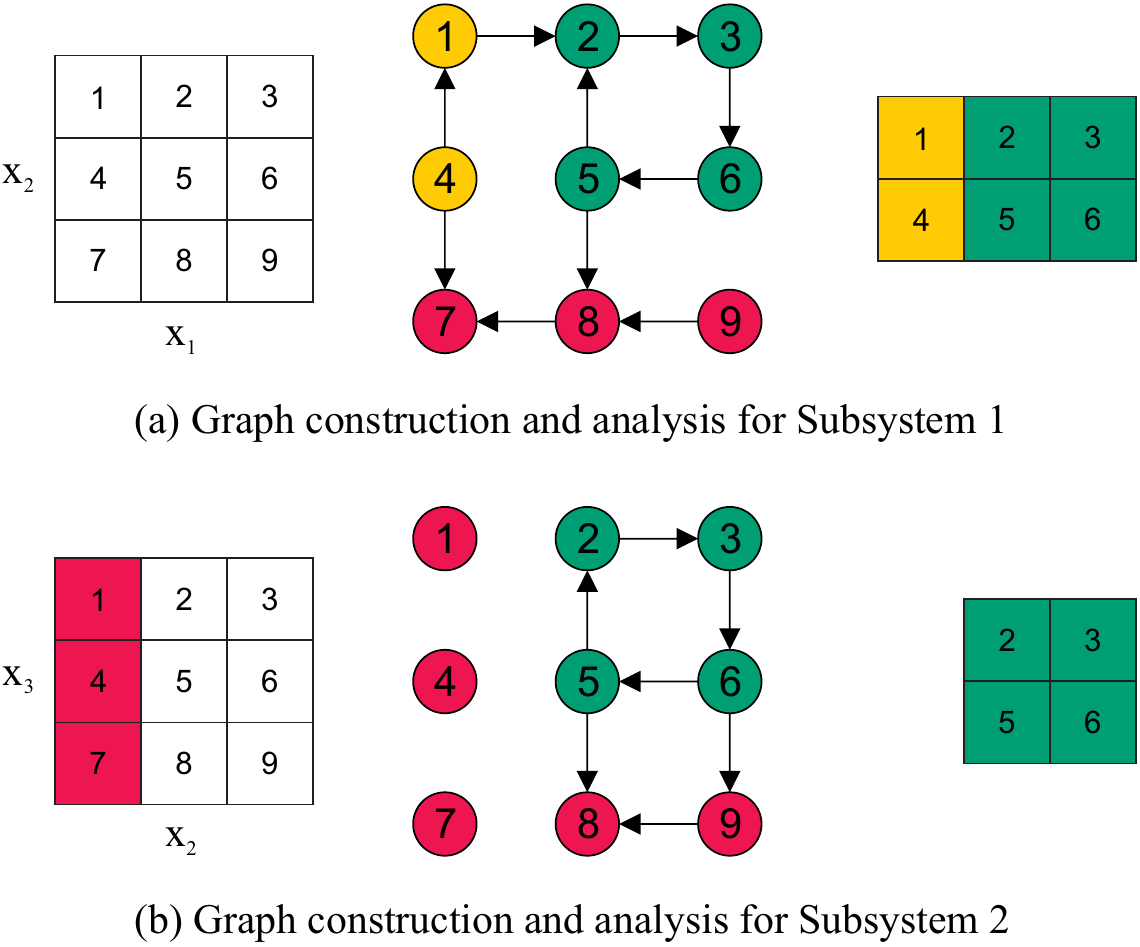}}
    \caption{\label{fig:distributed_graph_analysis} Procedure for the analysis of the distributed graphs to find the cells to be tested.}
\end{figure}

To illustrate how a cell is selected for testing, Figure~\ref{fig:distributed_graph_analysis} is presented. From the figure, it can be seen that the non-leaving cells of Subsystem~2 $R_2$ are $\{ B_2, \cdots, B_6\}$ while the leaving cells $\mathcal{C}^{lc}_2$  are $\{  B_8, B_9\}$. Also, the incoming neighbors of the leaving cells are $\{ B_5, B_6 \}$. Thus, $\mathcal{C}^{t}_{d,2} = \{ B_5, B_6 \}$. 

The algorithm for the set validation is presented in Algorithm~\ref{alg:set_reconstruction}. The algorithm takes as input the full system model, the constraint set, the cells to be tested $\mathcal{C}^t$ and the reconstructed solution $R$. It returns the final validated solution $R^*$. A while loop is used in the validation algorithm since the order to conduct the tests for the cells in the testing set in not known in advanced. Hence, we check several times until no cell is removed.

\begin{remark}
It is worth mentioning that the presence of both finite and infinite path passing through a particular cell may be caused by the effects of the actual system inputs and not the missing state. The effects of this two can be separated. However, in this work we take a conservative approach and treat the presence of both finite and infinite admissible path on a cell as an effect of the missing state.
\end{remark}

\begin{algorithm}

\caption{Set validation} \label{alg:set_reconstruction}

\KwIn{System~\eqref{eqn:system}, constraint sets $X$ and $U$, $\mathcal{C}^t$, $R$ }
\KwOut{$R^*$   }
$cont = true$ \\
\While{cont}{
    $tmp \leftarrow \emptyset$ \\
    \For{$B_i \in \mathcal{C}^t$}{

        \If{$F(B_i) \cap R = \emptyset$}{ 
            Add $B_i$ to $tmp$ \\     
        }

    }

    \If{$tmp \neq \emptyset$}{
        $R \leftarrow R \slash tmp$ \tcp{Set difference} 
        $\mathcal{C}^t \leftarrow \mathcal{C}^t \slash tmp$
    }

}
$R^* \leftarrow R$ \\

\textbf{return $R^*$}

\end{algorithm}

\subsection{Computational complexity}

In this section, we briefly analyze the computational complexity of the proposed algorithm. In particular, we focus on how the complexity of the algorithm increases in the system state dimension $n$. Consider that only two subsystems $S_1, S_2$ with state dimensions $n_1 < n,n_2 < n $ have been created and let $n_c$ denote the number of interval divisions per state dimension. Based on the results of \cite{decardi2021}, we know that the overall complexity of the graph-based control invariant set computation algorithm is determined by the graph construction step which is $\mathcal{O}(n_c^{n})$. This implies that the algorithm increases exponentially in the state dimension $n$.

Owing to the decomposition, the computation of the control invariant sets for the each of the subsystems is $\mathcal{O}(n_c^{n_1})$ and $\mathcal{O}(n_c^{n_2})$ respectively which is less than $\mathcal{O}(n_c^{n})$. Thus, the complexity of the decentralized and the distributed steps is determined by the subsystem with the largest dimension i.e. max($\mathcal{O}(n_c^{n_1}),\mathcal{O}(n_c^{n_2})$). In both of the steps, the missing states are determined and stored in a dictionary prior to the computation and therefore retrieving it has a complexity of $\mathcal{O}(1)$. The centralized step is probably the most time consuming part of the algorithm since the control invariant set is reconstructed from the subsystem solution and validated. However, since only the cells flagged for testing $\mathcal{C}^t$ are tested, three possibilities can occur depending on the number of cells flagged. Let $n_t$ be the number of cells flagged for testing after the distributed computation step. At worst, all the cells are flagged and therefore all the cells in the reconstruction step need to be validated. At the best case, no cell is flagged and therefore the overall control invariant set can be reconstructed without any validation. On the average case, not all the cells are flagged for validation. The complexity in this stage is therefore $\mathcal{O}(n_t)$ which is on average less than or at worst equal to $\mathcal{O}(n_c^{n})$. 

\section{Examples} \label{sec:examples}

In this section we present three numerical examples to demonstrate the efficacy of the proposed algorithm. The first two are 3 dimensional linear and nonlinear examples. This is used to demonstrate that the solution converges to the solution of the full system model for both the linear and nonlinear case. This is because, it is computationally difficult to obtain an approximation of $R_X$ for higher dimensional systems using the standard GIS algorithm. Finally, we compute an outer approximation of $R_X$ for a 6 dimensional nonlinear system. 

Unless otherwise stated, the number of division in each dimension was fixed at 128.The computation was performed on a laptop computer with Intel i7 CPU at 2.60 GHZ and 16 GB RAM. We refer to the solution from the decomposition as the distributed solution and that from the full system model as the centralized solution.

\subsection{Linear system example}

 Consider the linear time-invariant cascade system
 \begin{equation} \label{eqn:linear_example}
   x^+ = Ax + Bu
 \end{equation}
 where $A$ and $B$ are given by 
 \begin{equation*}
   A = \begin{bmatrix} 2 & 0 & 0 \\ 1 & 2 & 0 \\ 0 & 1 & 2 \end{bmatrix}  ~\text{and} ~  B =   \begin{bmatrix} 1 \\ 0 \\ 0 \end{bmatrix}.
 \end{equation*}
 The constraints on the states and input are $X = \{ x \in \mathbb{R}^3 : \| x \|_{\infty} \leq 5 \} $ and $U = \{ u \in \mathbb{R} : \| u \|_{\infty} \leq 1 \}$ respectively. 

 \begin{figure}[tbp] 
     \center{\includegraphics[width=0.7\columnwidth]{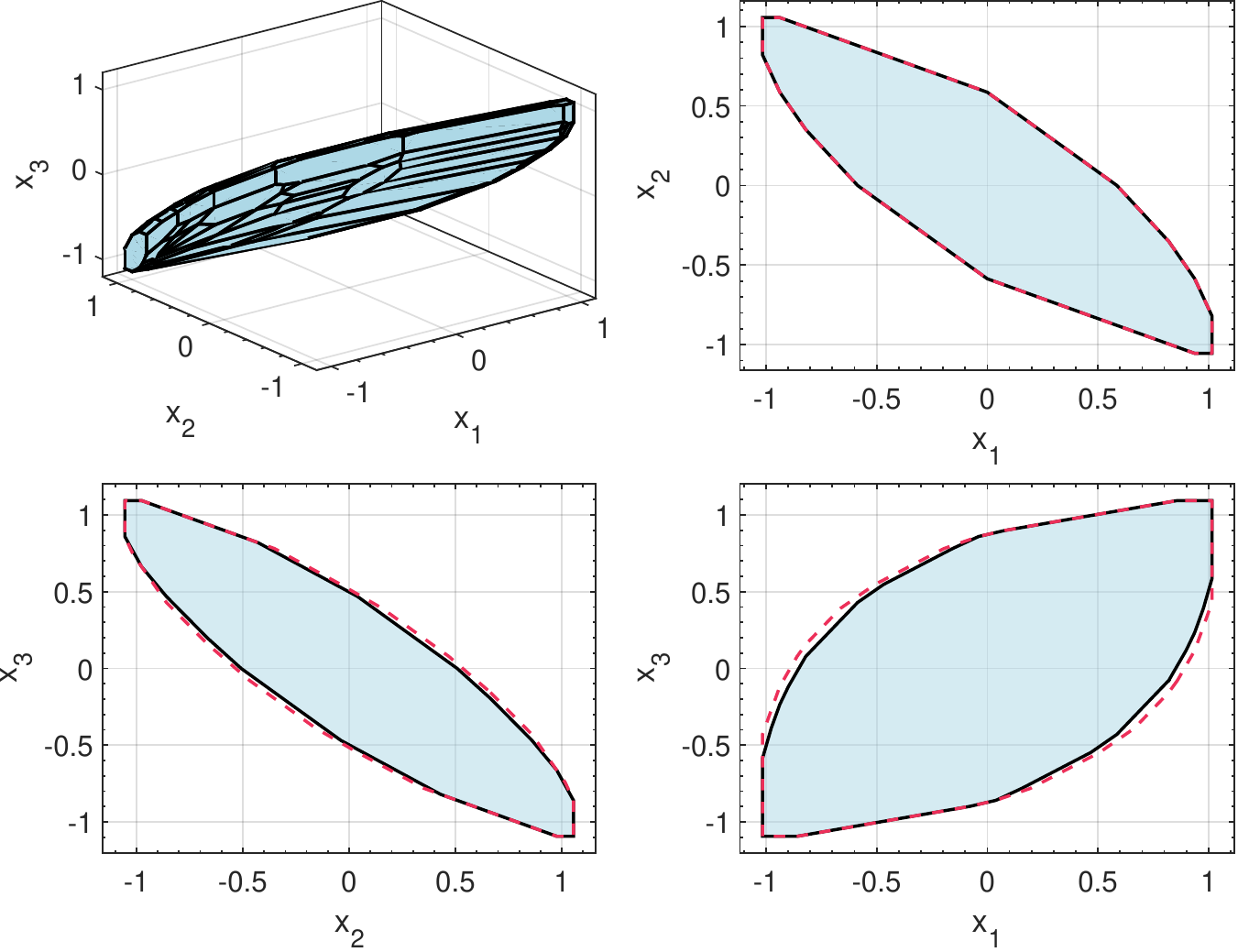}}
     \caption{\label{fig:example_1} Comparison of the convex hull of the cells in the centralized and distributed solutions. Blue shaded area: centralized solution. Dashed red: solution of the sets obtained from the distributed computation. Black: Final solution after reconstructing the solution and validating the cells. }
 \end{figure}
 Our goal is to compute the control invariant set for the system. To use the proposed algorithm, the system is first expanded to a 4-dimensional system by repeating the state equation for the second state as shown in Equation~(\ref{eqn:example1_expansion}). 

 \begin{equation} \label{eqn:example1_expansion}
     z^+
     = f(z,u) = 
     \left[
     \begin{array}{cc|cc}
     2 & 0 & 0 & 0 \\
     1 & 2 & 0 & 0 \\ \hline
     1 & 0 & 2 & 0 \\
     0 & 0 & 1 & 2 \\
     \end{array} \right]
     z + 
     \left[
     \begin{array}{c}
     1 \\
     0  \\ \hline
     0  \\
     0  \\
     \end{array} \right] u
 \end{equation} 
 where $z = [x_1 ~ x_2 ~ x_2 ~ x_3]^T \in \mathbb{R}^{4}$. It can be seen that $n_z = 4 > n_x = 3$. Equation (\ref{eqn:example1_expansion}) is then decomposed to obtain two subsystems namely: subsystem one: $z_1 = (x_1,x_2)$ and subsystem two: $z_2= (x_2,x_3)$ as shown in Equations (\ref{eqn:example1_subsystem1}) and (\ref{eqn:example2_subsystem2}) below.
 \begin{equation} \label{eqn:example1_subsystem1}
     z_1^+ = f_1(z_1, u)
     =
     \begin{bmatrix}
     2 & 0 \\
     1 & 2 
     \end{bmatrix} 
     z_1 + 
     \begin{bmatrix}
     1 \\
     0 
     \end{bmatrix} 
     u
     = f_1(x_1,x_2, u)
 \end{equation}

 \begin{equation}\label{eqn:example2_subsystem2}
     z_2^+ = f_2(z_2,x_1)
     =
     \begin{bmatrix}
     2 & 0 \\
     1 & 2 
     \end{bmatrix} 
     z_2 
     +
     \begin{bmatrix}
     1  \\
     0 
     \end{bmatrix} 
     x_1 +
     \begin{bmatrix}
     0 \\
     0 
     \end{bmatrix} 
     u
     = f_2(x_2,x_3,x_1)
 \end{equation}
 The control invariant sets for the two subsystems and ultimately the full system are computed using the distributed algorithm. The solution from the distributed algorithm is compared to the projections of the control invariant set computed using the centralized model. This is presented in Figure~\ref{fig:example_1}. It can be seen that the distributed computation converges to the centralized computation.  


\subsection{Nonlinear system example}

Consider the nonlinear system
\begin{subequations} \label{eqn:nonlinear_example}
    \begin{align} 
        x_1^+ = x_1^2 + u \\
        x_2^+ = x_2^2 + x_1 \\
        x_3^+ = x_3^2 + x_2
    \end{align}
\end{subequations}
The constraints on the states and input are $X = \{ x \in \mathbb{R}^3 : \| x \|_{\infty} \leq 5 \} $ and $U = \{ u \in \mathbb{R} : \| u \|_{\infty} \leq 1 \}$ respectively. 

Our goal is to compute an approximation of the largest control invariant set for the system. To use the proposed algorithm, the system is first expanded to a 4-dimensional system by repeating the state equation for the second state as shown in Equation (\ref{eqn:example2_expansion}).

\begin{equation} \label{eqn:example2_expansion}
    z^+
    = f(z,u) = 
    \left[
    \begin{array}{cc|cc}
    x_1^2 & 0 & 0 & 0 \\
    x_1 & x_2^2 & 0 & 0 \\ \hline
    x_1 & 0 & x_2^2 & 0 \\
    0 & 0 & x_2 & x_3^2 \\
    \end{array} \right]
     + 
    \left[
    \begin{array}{c}
    1.0 \\
    0  \\ \hline
    0  \\
    0  \\
    \end{array} \right] u
\end{equation} 
where $z = [x_1 ~ x_2 ~ x_2 ~ x_3]^T \in \mathbb{R}^{n_z}$ with $n_z=4 > n_x=3$. Thereafter, the expanded system is decomposed into two subsystems with overlapping states as described in the earlier sections.

The results from the computations using the full system model and that of the decomposed system model are presented in Figure~\ref{fig:example_2} and \ref{fig:computation_time}. Figure~\ref{fig:example_2} shows the convex hull of the sets in the distributed solution and the centralized solution. Again, it can be seen that the solution using the decomposed models converges to that of the full system model after the reconstruction. Furthermore, Figure~\ref{fig:computation_time} shows the scalability of the proposed distributed algorithm compared to the centralized algorithm. It can be seen that the algorithm using the decomposed model scales better than when the centralized model is used.

\begin{figure}[tbp] 
    \center{\includegraphics[width=0.7\columnwidth]{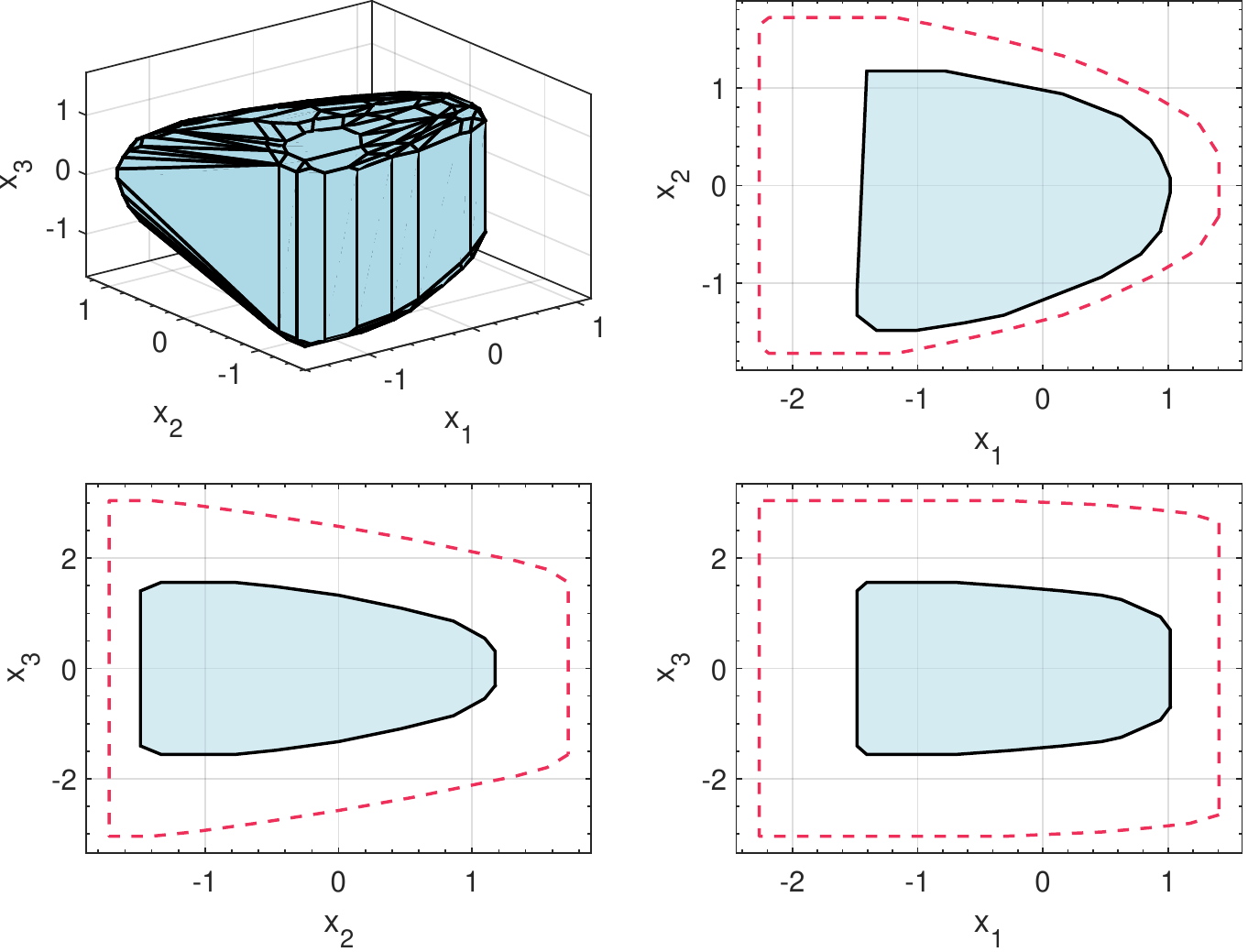}}
    \caption{\label{fig:example_2} Comparison of the convex hull of the cells in the solutions utilizing the decomposed model and the full system model. Blue shaded area: Solution from the computation utilizing the full system model. Dashed red: solution of the sets obtained from the distributed computation. Black: Final solution after reconstructing the solution and validating the cells. }
\end{figure}

\begin{figure}[tbp] 
    \center{\includegraphics[width=0.5\columnwidth]{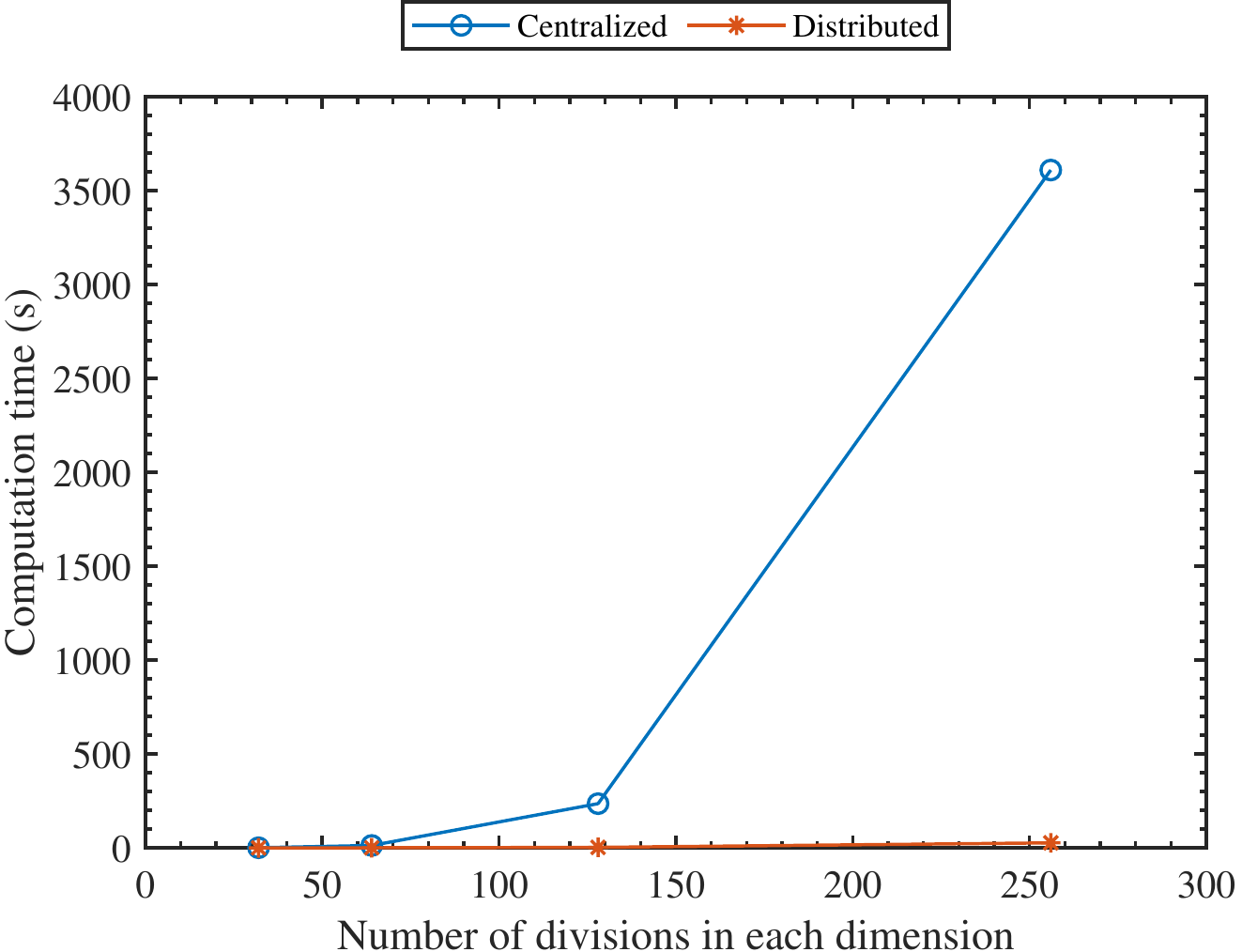}}
    \caption{\label{fig:computation_time} The computation vs the number of interval divisions in each dimension. At 256 interval divisions, it took more than an hour for the computation using the centralized model. An hour has been used for better visualization.}
\end{figure}

\subsection{Three Continuously stirred tank reactors in series example}

Let us consider an isothermal continuous-stirred tank reactor (CSTR) in which the following irreversible chained reactions occurs 
\begin{equation*}
    A \xrightarrow{k_1} B \xrightarrow{k_2} C
\end{equation*}
The reactor can be described by the following dimensionless modeling equations
\begin{subequations} \label{eqn:cstr_dimensionless_equations}
    \begin{align} 
        \frac{dx_1}{dt'} ={} & -x_1 + \text{Da}_1x_1 + u_1 \\
        \frac{dx_2}{dt'} ={} & - x_2 + \text{Da}_2 x_2^2 - \text{Da}_1 x_1 
    \end{align}
\end{subequations}
where Da$_1=1$ and Da$_2=2$ represent the dimensionless Damkholer number for the reactions 1 and 2, $x$ and $u$ are the dimensionless state and input vectors and $t'$ is dimensionless time. A similar version of the model was presented in \cite{Harmon1980}. By connecting three of the CSTR model in series (shown in Figure~\ref{fig:reactor_cascade}), the following six dimensional model is obtained.

\begin{figure}[tbp] 
    \center{\includegraphics[width=0.7\columnwidth]{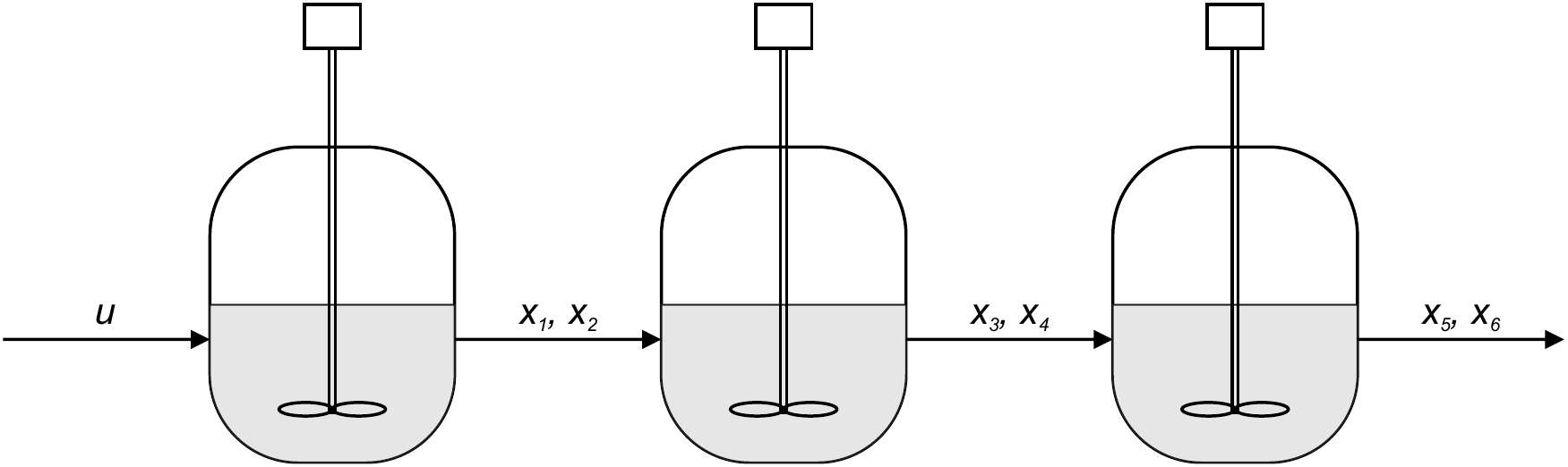}}
    \caption{\label{fig:reactor_cascade} A schematic diagram of the three CSTR in series.}
\end{figure}

\begin{subequations} \label{eqn:cstr_series}
    \begin{align} 
        \frac{dx_1}{dt'} ={} & -x_1 + \text{Da}_1x_1 + u_1 \\
        \frac{dx_2}{dt'} ={} & - x_2 + \text{Da}_2 x_2^2 - \text{Da}_1 x_1 \\
        \frac{dx_3}{dt'} ={} & -x_3 + \text{Da}_1x_3 + x_1 \\
        \frac{dx_4}{dt'} ={} & - x_4 + \text{Da}_2 x_4^2 - \text{Da}_1 x_3 + x_2\\
        \frac{dx_5}{dt'} ={} & -x_5 + \text{Da}_1x_5 + x_3 \\
        \frac{dx_6}{dt'} ={} & - x_6 + \text{Da}_2 x_6^2 - \text{Da}_1 x_5 + x_4
    \end{align}
\end{subequations}
The continuous-time model is discretized using a step size of 1 before usage in GIS algorithm. The constraints on the states and input are $X = \{ x \in \mathbb{R}^6 : 0 \leq x  \leq 1 \} $ and $U = \{ u \in \mathbb{R} : 0 \leq u  \leq 1 \}$ respectively. The results of each reactor is presented in Figures~\ref{fig:example_3_r1}--\ref{fig:example_3_r3}. The computations were performed on a computing server with single core Intel Xeon E5 processor with 2.1 GHz frequency and 100 GB of RAM. It took roughly 4 hours to compute the solution using the distributed approach with the reconstruction and validation step taking about 70 \% of the total computation time. It is worth mentioning that it is intractable to compute an approximation of the largest control invariant set for the six dimensional system using the standard GIS algorithm without system decomposition.
\begin{figure}[tbp] 
    \center{\includegraphics[width=0.5\columnwidth]{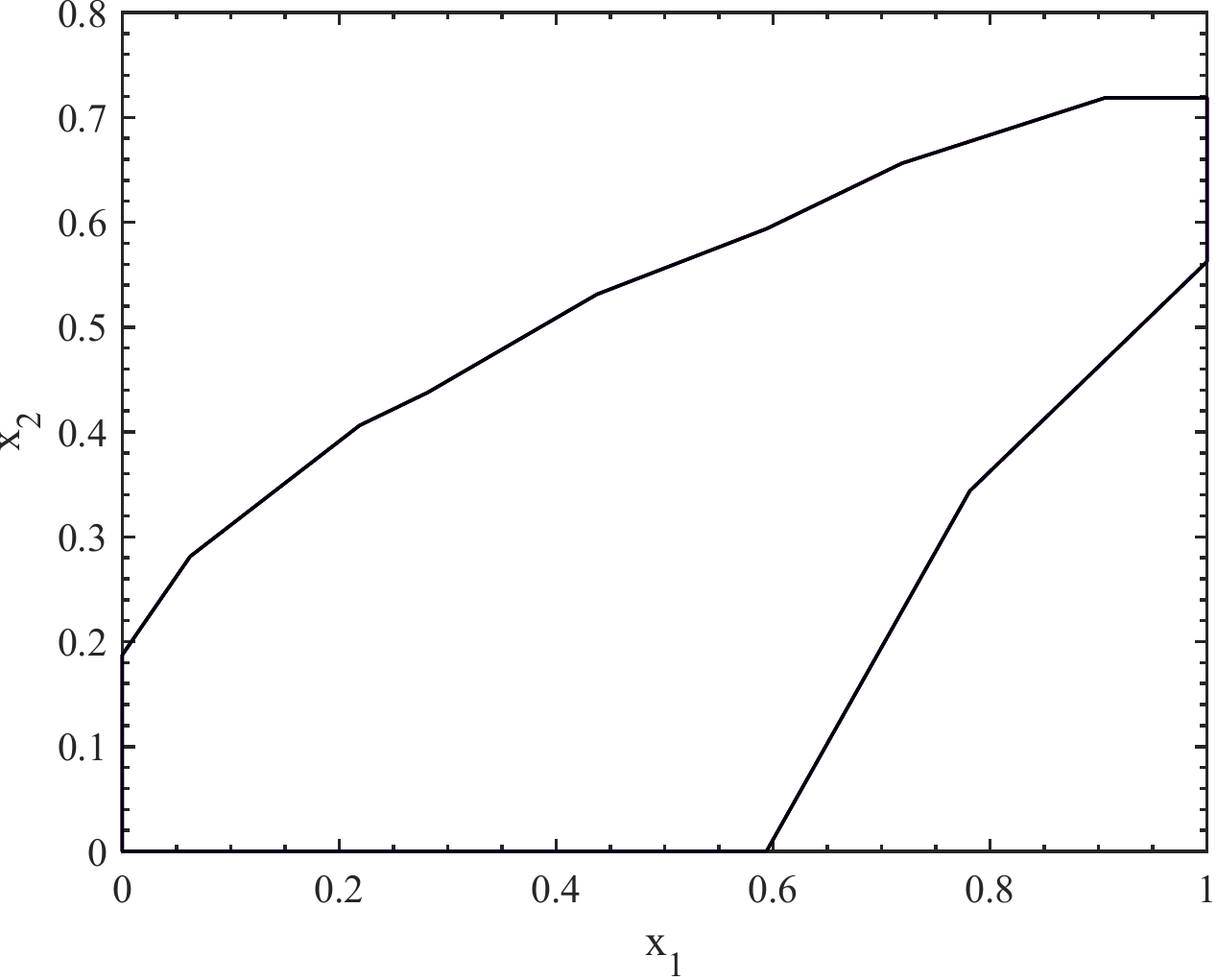}}
    \caption{\label{fig:example_3_r1} The solution for Reactor 1. }
\end{figure}

\begin{figure}[tbp] 
    \center{\includegraphics[width=0.5\columnwidth]{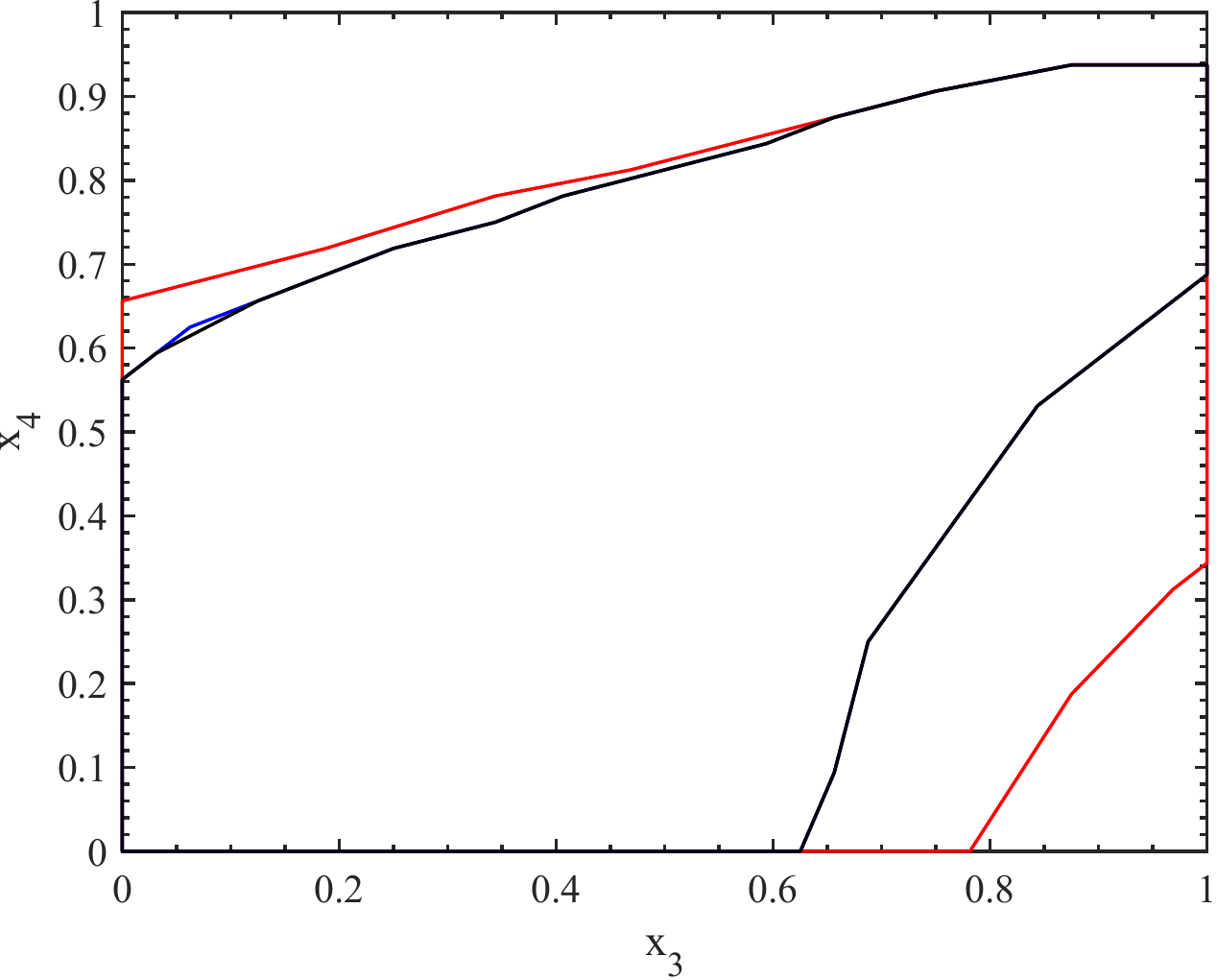}}
    \caption{\label{fig:example_3_r2} The solution for Reactor 2. Red: results from decentralized computation. Blue: results from the distributed computation. Black: Final solution after validation}
\end{figure}

\begin{figure}[tbp] 
    \center{\includegraphics[width=0.5\columnwidth]{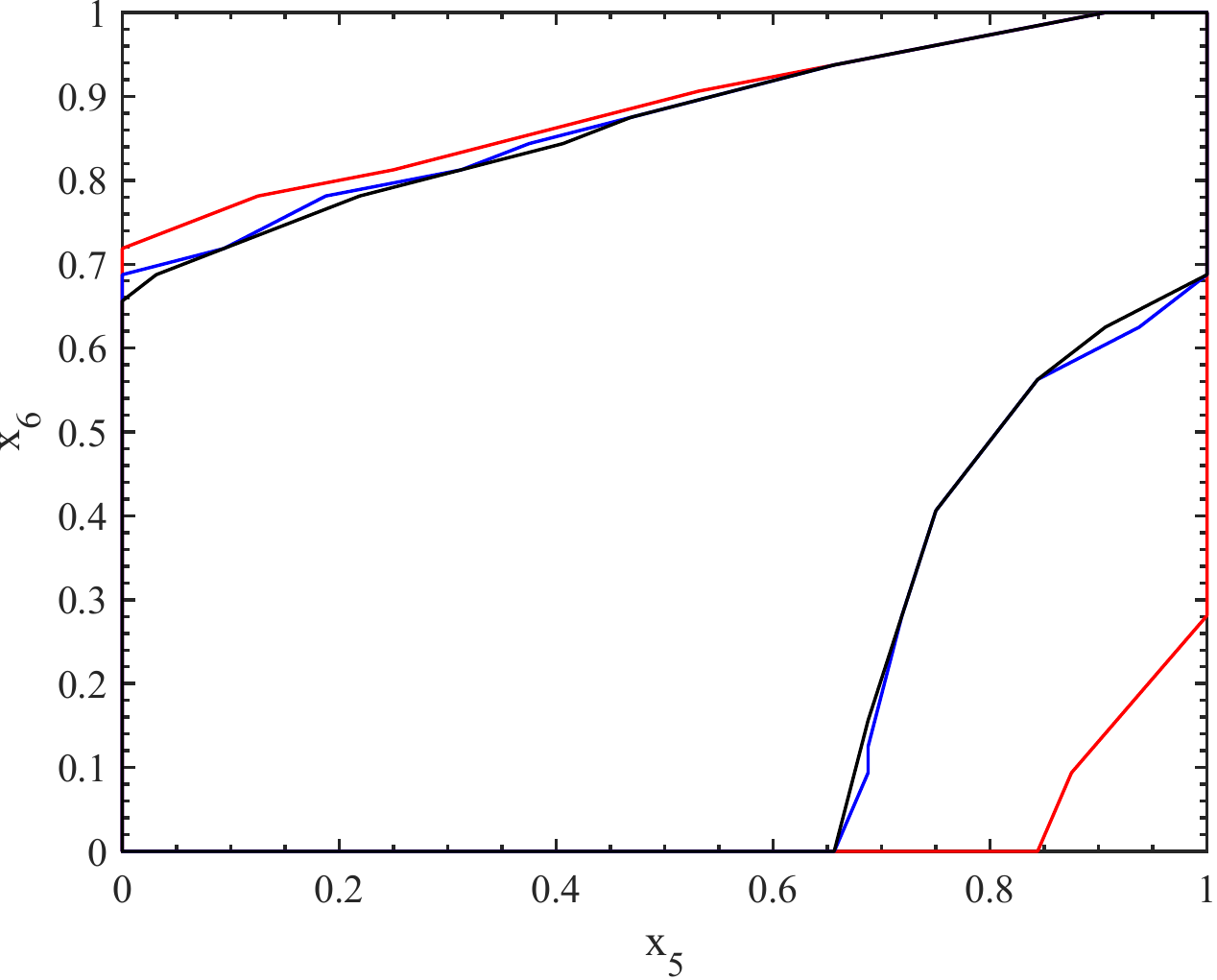}}
    \caption{\label{fig:example_3_r3} The solution for Reactor 3. Red: results from decentralized computation. Blue: results from the distributed computation. Black: Final solution after validation}
\end{figure}

\section{Concluding remarks} \label{sec:concluding_remarks}

Given a constrained discrete-time time-invariant control system, this paper obtains outer approximations of the largest control invariant set contained in the state constraint. For this purpose, we presented a graph-based distributed algorithm which approximates the dynamics of the control system as a directed graph allowing for analysis of the system using graph theory. This algorithm paves a promising way to overcome the ``curse of dimensionality'' encountered in control invariant set computation. Simulations using several numerical examples demonstrate the ability and effectiveness of the proposed method to approximate the largest control invariant set, just like the centralized algorithm. 

Future work will focus on further improving the computational efficiency of the proposed algorithm by utilizing recent advances in interdependent graph networks to analyze the distributed graphs for control invariance. This way the validation step can be conducted without having to reconstruct the set from the subsystem solutions. Another research direction is appropriate representation of the set obtained from the algorithm for MPC applications. The use of machine learning techniques will be explored.


\begin{thebibliography}{10}

\bibitem{homer2017}
T.~Homer and P.~Mhaskar, ``Constrained control lyapunov function-based control
  of nonlinear systems,'' {\em Systems \& Control Letters}, vol.~110,
  pp.~55--61, 2017.

\bibitem{mitchell2005}
I.~Mitchell, A.~Bayen, and C.~Tomlin, ``A time-dependent hamilton-jacobi
  formulation of reachable sets for continuous dynamic games,'' {\em {IEEE}
  Transactions on Automatic Control}, vol.~50, pp.~947--957, jul 2005.

\bibitem{homer2020}
T.~Homer, M.~Mahmood, and P.~Mhaskar, ``A trajectory-based method for
  constructing null controllable regions,'' {\em International Journal of
  Robust and Nonlinear Control}, vol.~30, no.~2, pp.~776--786, 2020.

\bibitem{aubin2009}
J.~P. Aubin, {\em Viability theory}.
\newblock Modern {Birkhäuser} classics, Boston: Birkhäuser, 2009.

\bibitem{blanchini1999}
F.~Blanchini, ``Set invariance in control,'' {\em Automatica}, vol.~35, no.~11,
  pp.~1747--1767, 1999.

\bibitem{mayne2001}
D.~Q. Mayne, ``Control of constrained dynamic systems,'' {\em European Journal
  of Control}, vol.~7, no.~2-3, pp.~87--99, 2001.

\bibitem{cannon2003}
M.~Cannon, V.~Deshmukh, and B.~Kouvaritakis, ``Nonlinear model predictive
  control with polytopic invariant sets,'' {\em Automatica}, vol.~39, no.~8,
  pp.~1487--1494, 2003.

\bibitem{decardi2021}
B.~Decardi-Nelson and J.~Liu, ``Computing robust control invariant sets of
  constrained nonlinear systems: A graph algorithm approach,'' {\em Computers
  \& Chemical Engineering}, vol.~145, p.~107177, 2021.

\bibitem{rakovic2010}
S.~V. Rakovi{\'c}, B.~Kern, and R.~Findeisen, ``Practical set invariance for
  decentralized discrete time systems,'' in {\em 49th IEEE Conference on
  Decision and Control (CDC)}, pp.~3283--3288, IEEE, 2010.

\bibitem{rungger2017}
M.~{Rungger} and P.~{Tabuada}, ``Computing robust controlled invariant sets of
  linear systems,'' {\em IEEE Transactions on Automatic Control}, vol.~62,
  pp.~3665--3670, July 2017.

\bibitem{fiacchini2010}
M.~Fiacchini, T.~Alamo, and E.~Camacho, ``On the computation of convex robust
  control invariant sets for nonlinear systems,'' {\em Automatica}, vol.~46,
  pp.~1334--1338, Aug. 2010.

\bibitem{bravo2005}
J.~M. Bravo, D.~Lim{\'o}n, T.~Alamo, and E.~F. Camacho, ``On the computation of
  invariant sets for constrained nonlinear systems: An interval arithmetic
  approach,'' {\em Automatica}, vol.~41, no.~9, pp.~1583--1589, 2005.

\bibitem{bertsekas1972}
D.~{Bertsekas}, ``Infinite time reachability of state-space regions by using
  feedback control,'' {\em IEEE Transactions on Automatic Control}, vol.~17,
  pp.~604--613, October 1972.

\bibitem{homer2018}
T.~Homer and P.~Mhaskar, ``Utilizing null controllable regions to stabilize
  input-constrained nonlinear systems,'' {\em Computers \& Chemical
  Engineering}, vol.~108, pp.~24--30, 2018.

\bibitem{kurzhanski2000}
A.~B. Kurzhanski and P.~Varaiya, ``Ellipsoidal techniques for reachability
  analysis: internal approximation,'' {\em Systems \& control letters},
  vol.~41, no.~3, pp.~201--211, 2000.

\bibitem{frehse2011}
G.~Frehse, C.~Le~Guernic, A.~Donz{\'e}, S.~Cotton, R.~Ray, O.~Lebeltel,
  R.~Ripado, A.~Girard, T.~Dang, and O.~Maler, ``Spaceex: Scalable verification
  of hybrid systems,'' in {\em International Conference on Computer Aided
  Verification}, pp.~379--395, Springer, 2011.

\bibitem{chen2018}
M.~Chen, S.~L. Herbert, M.~S. Vashishtha, S.~Bansal, and C.~J. Tomlin,
  ``Decomposition of reachable sets and tubes for a class of nonlinear
  systems,'' {\em IEEE Transactions on Automatic Control}, vol.~63, no.~11,
  pp.~3675--3688, 2018.

\bibitem{riverso2017}
S.~Riverso, K.~Kouramas, and G.~Ferrari-Trecate, ``Decentralized and
  distributed robust control invariance for constrained linear systems,'' in
  {\em 2017 IEEE 56th Annual Conference on Decision and Control (CDC)},
  pp.~5978--5984, IEEE, 2017.

\bibitem{li2020}
A.~Li and M.~Chen, ``Guaranteed-safe approximate reachability via state
  dependency-based decomposition,'' in {\em 2020 American Control Conference
  (ACC)}, pp.~974--980, IEEE, 2020.

\bibitem{decardi2022ACC}
B.~Decardi-Nelson and J.~Liu, ``A distributed control invariant set computing
  algorithm for nonlinear cascade systems,'' in {\em Proceedings of the
  American Control Conference}, accepted, 2022.

\bibitem{kerrigan2001}
E.~C. Kerrigan, {\em Robust constraint satisfaction: Invariant sets and
  predictive control}.
\newblock PhD thesis, University of Cambridge, 2001.

\bibitem{osipenko2007}
G.~Osipenko, {\em Dynamical Systems, Graphs, and Algorithms}.
\newblock No.~1889 in Lecture Notes in Mathematics, Berlin ; New York:
  {Springer}, 2007.
\newblock OCLC: ocm75927357.

\bibitem{TROTTER1978}
W.~T. Trotter and P.~Erdös, ``When the cartesian product of directed cycles is
  hamiltonian,'' {\em Journal of Graph Theory}, vol.~2, no.~2, p.~137–142,
  1978.

\bibitem{KAYNAMA2013}
S.~Kaynama and M.~Oishi, ``A modified riccati transformation for decentralized
  computation of the viability kernel under {LTI} dynamics,'' {\em IEEE
  Transactions on Automatic Control}, vol.~58, p.~2878–2892, Nov 2013.

\bibitem{Harmon1980}
W.~Harmon~Ray, {\em Advanced Process Control}.
\newblock Chemical Engineering S., New York, NY: McGraw-Hill, Sept. 1980.

\end{thebibliography}

\end{document}